\documentclass{optica-article}

\usepackage{lineno}
\usepackage{subcaption}
\usepackage{algorithm}
\usepackage{algpseudocode}

\journal{opticajournal} 

\articletype{Research Article}

\newtheorem{remark}{Remark}
\newtheorem{lemma}{Lemma}
\newtheorem{theorem}{Theorem}
\newtheorem{corollary}{Corollary}
\newtheorem{definition}{Definition}

\newenvironment{proof}{\noindent\textit{Proof.}}{$\null\nobreak\hfill\ensuremath{\square}$\vspace{2mm}}

\begin{document}

\title{Holographic phase retrieval via Wirtinger flow: Cartesian form with auxiliary amplitude}

\author{Ittetsu Uchiyama, Chihiro Tsutake,\authormark{*} Keita Takahashi, and Toshiaki Fujii}

\address{Department of Information and Communication Engineering, Nagoya University, Furo-cho, Chikusa-ku, Nagoya, 464-8603, Japan}

\email{\authormark{*}ctsutake@nagoya-u.jp} 


\begin{abstract*}
We propose a new gradient method for holography,
where a phase-only hologram is parameterized by not only the phase but also amplitude.
The key idea of our approach is the formulation of a phase-only hologram using an auxiliary amplitude.
We optimize the parameters using the so-called Wirtinger flow algorithm in the Cartesian domain,
which is a gradient method defined on the basis of the Wirtinger calculus.
At the early stage of optimization, 
each element of the hologram exists inside a complex circle,
and it can take a large gradient while diverging from the origin.
This characteristic contributes to accelerating the gradient descent.
Meanwhile, at the final stage of optimization, each element evolves along a complex circle,
similar to previous state-of-the-art gradient methods.
The experimental results demonstrate that our method outperforms previous methods,
primarily due to the optimization of the amplitude.
\end{abstract*}

\section{Introduction}
\label{s1}
Holography~\cite{Gabor48,Denisyuk62,Lippmann91}, 
a method for displaying 3D images on the basis of optical interference and diffraction,
has gained significant attention in a variety of research fields~\cite{Javidi21}.
Recently, holography has been implemented using a spatial light modulator~(SLM)~\cite{Yang23},
which can display 3D images by modulating either the amplitude or phase of incident light.
The quality of the displayed images primarily depends on the amplitude- or phase-modulation pattern, 
referred to as computer-generated hologram~\cite{Sahin20}.
Therefore, computing accurate holograms for displaying high-quality images is an important task in holography.
Throughout this paper,
we assume that an SLM can modulate only the phase of incident light, not the amplitude.
Under such assumption, we address the so-called holographic phase retrieval problem:
optimization of the phase-only hologram, 
i.e., phase-modulation pattern, while keeping the amplitude constant.

To develop our theory,
we begin with a brief overview of previous methods for holographic phase retrieval.
In the Gerchberg-Saxton~(GS) method~\cite{Gerchberg72}, 
a phase-only hologram is obtained through iterative projections between the image and hologram constraints.
Variants of the GS method have been proposed by many researchers, 
e.g., Fienup~\cite{Fienup82} and Netrapalli et al.~\cite{Netrapalli15}.
Guendy et al.~\cite{Guendy21} applied the Kaczmarz method~\cite{Kaczmarz37} to holography,
where a quadratic system of equations corresponding to holographic phase retrieval is approximately solved.
Chakravarthula et al.~\cite{Chakravarthula19} proposed a gradient method for holography,
where a hologram is obtained by minimizing a loss function defined in the image domain.
In \cite{Zhang22}, deep learning methods for holography are summarized;
a hologram is inferred by neural networks that are optimized on a vast set of images and/or holograms.
In \cite{Wang24color}, Wang et al. proposed a chromatic aberration-free method for color holography with large viewing angle;
in \cite{Wang24liquid}, they also proposed an optical method,
combining a liquid-lens-based camera and physical model-driven neural networks.

Among the aforementioned previous works,
the gradient method produces holograms that have the capability to display high-quality images.
In the gradient method~\cite{Chakravarthula19} and its variants~\cite{Chakravarthula20,Peng20,Wang22},
a phase-only hologram is parameterized by the phase only, not the amplitude.
During optimization,
each element of a hologram is updated \textbf{along a complex circle}, 
which is an important characteristic of the gradient methods.
Theoretically, 
optimizing only the phase is necessary and sufficient for a phase-only hologram;
in other words, parameterization and optimization of the amplitude are not required.
This is because, even if the amplitude is employed as a parameter,
it does not affect the quality of displayed images.
However, 
we believe that exclusively focusing on the phase poses a major barrier to accelerating the gradient descent and reaching the optimal hologram.

On the basis of the aforementioned discussion, we propose a new gradient method, 
where a phase-only hologram is parameterized by not only the phase but also amplitude.
The key idea of our approach is the formulation of a phase-only hologram using an \textit{auxiliary amplitude}.
For optimization,
we transform the phase and amplitude to complex parameters in Cartesian form.
We optimize the parameters using the Wirtinger flow algorithm~\cite{Chen15,Candes15,Zhang17},
a gradient method defined on the basis of the Wirtinger calculus~\cite{Wirtinger27}.
At the early stage of optimization, 
each element of the hologram exists \textbf{inside a complex circle},
and it can take a large gradient while diverging from the origin.
This characteristic contributes to accelerating the gradient descent.
Meanwhile, at the final stage of optimization, each element evolves along a complex circle,
similar to the previous gradient methods.
The experimental results demonstrate that our method outperforms previous methods,
primarily due to the optimization of the amplitude.

\subsection*{Paper organization}
In Section~\ref{s2}, 
we formulate the previous gradient method~\cite{Chakravarthula19,Chakravarthula20,Peng20,Wang22},
hereafter referred to as \textit{Wirtinger flow in polar form~(WFPF)} for convenience.
In Section~\ref{s3}, we propose our method, referred to as \textit{Wirtinger flow in Cartesian form~(WFCF)}, 
and present its mathematical properties.
In Section~\ref{s4}, we report the experimental results,
and in Section~\ref{s5}, we conclude with a brief summary.
Most of the mathematical proofs are presented in \ref{a1} and \ref{a2}.

\subsection*{Notations}
The imaginary unit is denoted by $j$, satisfying $j^2 = -1$.
A scalar variable and a vector variable are denoted by regular typeface and boldface, respectively, 
e.g., $c$ and $\boldsymbol{c}$.
The amplitude and phase of a complex scalar $c$ are denoted by $|c|$ and $\mathop{\mathrm{arg}}(c)$, respectively.
We denote the conjugates of a complex scalar $c$ and a complex vector $\boldsymbol{c}$ by $\bar{c}$ and $\bar{\boldsymbol{c}}$, respectively.
\section{Wirtinger flow in polar form}
\label{s2}
We aim to display a planar image on a screen that is parallel to a hologram~(SLM) plane.
Figure~\ref{fig:coordinate} shows the coordinate system.
\begin{figure}[!h]
\centering
\includegraphics[scale=1.0]{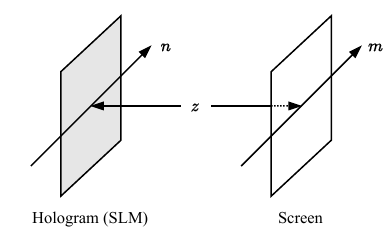}
\caption{Coordinate system.}
\label{fig:coordinate}
\end{figure}
\newline
We denote the distance between the two planes by $z$.
Without loss of generality, 
we consider a 1D hologram and 1D image with the sizes $N$ and $M$, respectively.
All the discussions are extendable to a 2D or 3D setup by increasing the dimensions of the hologram and image.

Let $n \in [0, N)$ be the coordinate on the hologram plane, 
$h_n$ be the $n$-th element of a phase-only hologram, and $\phi_n \in [0, 2\pi)$ be the phase of $h_n$.
The $n$-th element $h_n$ is defined as the following complex number in polar form:
\begin{equation}
h_n(\phi_n) = A\exp(j \phi_n),
\label{eq:hol-p}
\end{equation}
where $A > 0$ is a constant amplitude, not the parameter to be optimized.
We hereafter use the unit amplitude $A=1$.
We can interpret $A\exp(j \phi_n)$ as a complex variable whose length is normalized to unity.
Therefore, we refer to $A\exp(j \phi_n)$ as the \textit{normalized hologram} for convenience.

Let $\boldsymbol{\phi} \in [0, 2\pi)^N$ be a vector formed by stacking all $\phi_n$, 
and $m \in [0,M)$ be the coordinate on the screen plane.
The wavefront propagation between the two planes is modeled as the following Rayleigh-Sommerfeld diffraction~\cite{Goodman05} in a discrete setup:
\begin{equation}
p_m(\boldsymbol{\phi})=\sum_{n=0}^{N-1} w_{mn} h_n(\phi_n),
\label{eq:prop}
\end{equation}
where $w_{mn}$ is a propagation kernel.
For the limit distance $z=\infty$, 
note that the kernel is identical to the Fourier basis function~\cite{Goodman05}.
This characteristic is leveraged in our analysis~(Theorem~\ref{t2} in Section~\ref{s3ss2}).

Let $|i_m|^2$ be the $m$-th light magnitude of an image we aim to display.
In the WFPF,
holographic phase retrieval is defined as the recovery of unknown $\boldsymbol{\phi}$ from the quadratic system
\begin{equation} 
\mathop{\mathrm{find}}_{\boldsymbol{\phi} \in [0, 2\pi)^N} \quad
|i_m|^2 = |p_m(\boldsymbol{\phi})|^2.
\label{eq:prpf}
\end{equation}
Only the magnitudes are perceived by observes and thus considered here.
To solve \eqref{eq:prpf} using the gradient descent~\cite{Boyd04},
it is transformed into the empirical minimization problem
\begin{equation}
\min_{\boldsymbol{\phi} \in [0, 2\pi)^N} \quad
L(\boldsymbol{\phi}) = \sum_{m=0}^{M-1} 
\bigg(|p_m(\boldsymbol{\phi})|^2 - |i_m|^2\bigg)^2.
\label{eq:losspf}
\end{equation}
Let $\partial / \partial \phi_n$ be the Wirtinger derivative operator with respect to $\phi_n$,
see \cite{Chakravarthula19} for more details.
The minimization problem~\eqref{eq:losspf} is solved using the gradient descent
\begin{equation}
\phi_n[\tau + 1] = 
\phi_n[\tau] - \alpha 
\frac{\partial}{\partial \phi_n}L(\boldsymbol{\phi}[\tau]),
\label{eq:WFPF}
\end{equation}
where $\tau$ is the current number of iterations, and $\alpha$ is a learning rate.
We refer to \eqref{eq:WFPF} as the WFPF because the hologram in \eqref{eq:hol-p} is defined in polar form.

Figure~\ref{fig:flow_a} illustrates a geometrical interpretation of the gradient $\partial L / \partial \phi_n$.
Throughout the iterations,
the phase $\phi_n$ is updated while keeping the amplitude constant $A=1$.
In other words, 
each element of the normalized hologram $A\exp(j\phi_n)$ evolves along a complex circle of unit radius,
which is an important characteristic of the WFPF.
We believe that optimizing only the phase is an obstacle to accelerating the gradient descent and reaching the optimal hologram.
\begin{figure}[!t]
\centering
\begin{subfigure}{0.48\textwidth}
\centering
\includegraphics[scale=0.98]{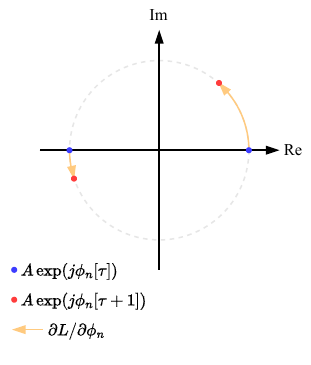}
\caption{WFPF~\cite{Chakravarthula19,Chakravarthula20,Peng20,Wang22}}
\label{fig:flow_a}
\end{subfigure}
\begin{subfigure}{0.48\textwidth}
\centering
\includegraphics[scale=0.98]{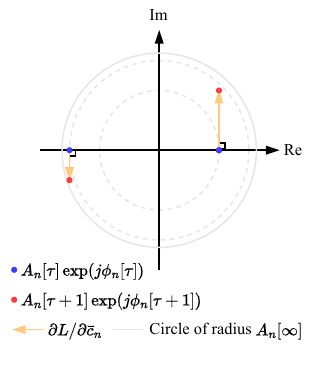}
\caption{WFCF~(ours)}
\label{fig:flow_b}
\end{subfigure}
\caption{Geometry of gradients.}
\label{fig:flow}
\end{figure}
\section{Wirtinger flow in Cartesian form}
\label{s3}
\subsection{Formulation}
\label{s3ss1}
In contrast to the WFPF, 
we parameterize a phase-only hologram by not only the phase but also amplitude.
To this end, we formulate a new hologram using an auxiliary amplitude,
which plays a central role in our method.
Let $A_n > 0$ be the auxiliary amplitude assigned for each $n$, 
which is the parameter to be optimized in our method.
We substitute the following trivial relationships into the previous hologram~\eqref{eq:hol-p}:
\begin{equation}
\left\{\,
\begin{aligned}
A &= 1 = \frac{A_n}{A_n}\\
\exp(j\phi_n) &= 
\frac{\hspace{0.5mm}\exp(j\phi_n)}{|\exp(j\phi_n)|}   
\end{aligned}
\right.,
\end{equation}
where we used $|\exp(j\phi_n)|=1$.
As a result, we can obtain the following new hologram defined in polar form:
\begin{equation}
h_n(A_n,\phi_n) = \frac{A_n\exp(j\phi_n)}{|A_n\exp(j\phi_n)|}.
\label{eq:hol-our}
\end{equation}
Note that $A_n\exp(j\phi_n)$ in \eqref{eq:hol-our} can represent any point on the Gaussian plane,
while $A\exp(j\phi_n)$ in \eqref{eq:hol-p} is always constrained to the complex unit circle.
The difference between the WFPF and WFCF can be visually interpreted as in Fig.~\ref{fig:flow}.
More formal discussion will be presented in Section~\ref{s3ss2}.
In contrast to the normalized hologram, 
we refer to $A_n\exp(j\phi_n)$ as the \textit{non-normalized hologram} as each element likely has a non-unit amplitude.

We optimize the amplitude and phase using the Wirtinger flow algorithm~\cite{Chen15,Candes15,Zhang17},
a gradient method for complex variables defined in the Cartesian domain.
To this end, we transform the polar-form hologram in \eqref{eq:hol-our} to a Cartesian-form one.
Let $a_n$ and $b_n$ be real variables assigned for each $n$, and $c_n=a_n+jb_n$ be a complex variable.
Replacing $A_n $ and $\phi_n$ in \eqref{eq:hol-our} with $|c_n|$ and $\mathop{\mathrm{arg}}(c_n)$, respectively,
we derive the Cartesian-form hologram
\begin{equation}
h_n(c_n)=\frac{c_n}{|c_n|} \quad (|c_n| > 0),
\label{eq:hol-c}
\end{equation}
where we used $c_n = |c_n|\exp(j \mathop{\mathrm{arg}}(c_n))$.
When $c_n$ is zero, the zero division occurs in \eqref{eq:hol-c};
therefore, we must guarantee $|c_n| > 0$.
Note that whereas $c_n$ can take any complex number in the Gaussian plane,
$h_n(c_n)$ is always constrained to the unit circle, just like $h_n(\phi_n)$.
Therefore, $h_n(c_n)$ can be a solution to holographic phase retrieval.

Let $\boldsymbol{c} \in \mathbb{C}^N$ be a complex vector formed by stacking all $c_n$.
Replacing $\phi_n$ and $\boldsymbol{\phi}$ in \eqref{eq:prop}--\eqref{eq:losspf} with $c_n$ and $\boldsymbol{c}$, respectively,
we derive the empirical minimization problem
\begin{equation}
\min_{\boldsymbol{c}\in\mathbb{C}^N} \quad
L(\boldsymbol{c}) = \frac{1}{8N^2}\sum_{m=0}^{M-1} 
\bigg(|p_m(\boldsymbol{c})|^2 - |i_m|^2\bigg)^2.
\label{eq:loss}
\end{equation}
We introduced the numerical constant $1/8N^2$ to simplify our theoretical results in Section~\ref{s3ss2}.
When we implement our method on a computer, the constant is ignored.
We would like to find the minimizer of \eqref{eq:loss} via the gradient descent~\cite{Boyd04},
but the loss function $L(\boldsymbol{c})$ is non-holomorphic~\cite{Lang85} with respect to $c_n$.
We thus employ the Wirtinger calculus~\cite{Wirtinger27},
which enables us to compute the derivative of $L(\boldsymbol{c})$.
\begin{definition}[Wirtinger derivative]
For a scalar function $f(\boldsymbol{c})$ or $f(c_n)$, 
the Wirtinger derivative with respect to $\bar{c}_n$ is defined as follows:
\begin{align}
\frac{\partial}{\partial \bar{c}_n}f &= \frac{1}{2}
\bigg(\frac{\partial}{\partial a_n}f + j \frac{\partial}{\partial b_n}f\bigg).
\end{align}
\end{definition}
\noindent
We find the solution to \eqref{eq:loss} using the following gradient descent, also known as the Wirtinger flow~\cite{Chen15,Candes15,Zhang17}:
\begin{equation}
c_n[\tau + 1] = 
c_n[\tau] - \alpha \frac{\partial}{\partial \bar{c}_n}L(\boldsymbol{c}[\tau]).
\label{eq:WFCF}
\end{equation}
We refer to \eqref{eq:WFCF} as the WFCF because the hologram in \eqref{eq:hol-c} is defined in Cartesian form.
The concrete formula of $\partial L/\partial \bar{c}_n$ is presented in \eqref{eq:grad}.
The formulation of our hologram in \eqref{eq:hol-c} shows that
we must guarantee $|c_n[\tau]| \neq 0$  at any $\tau$, which appears to be challenging.
However, we will show in Remark~\ref{r4} in Section~\ref{s3ss2} that 
$c_n[\tau]$ never degenerates to the zero point,  
regardless of the initial value $c_n[0]$ and the propagation distance $z$.

Algorithm~\ref{alg:WFPF} shows the overall implementation of the WFCF for the limit distance $z=\infty$.
A learning rate and the maximum number of iterations are hyper-parameters,
which will be empirically chosen in Section~\ref{s4}.
The gradients are evaluated by computing  $\mathsf{tmp1}$, $\mathsf{tmp2}$, and $\mathsf{tmp3}$.
These temporal variables are computed via the forward and inverse fast Fourier transformation,
denoted by $\mathsf{fft}$ and $\mathsf{ifft}$, respectively.

\begin{algorithm}[!t]
\caption{Wirtinger flow in Cartesian form for limit distance $z=\infty$}
\label{alg:WFPF}
\begin{algorithmic}
\Require Image $|i_m|^2$, initial hologram $\boldsymbol{c}[0]$, learning rate $\alpha$, and maximum number of iterations $T$
\Ensure $\boldsymbol{c}[T-1]$
\For{$\tau=0$ \textbf{to} $T-1$}
\State $\mathsf{tmp1} = \mathsf{fft}(\boldsymbol{c}[\tau])$
\State $\mathsf{tmp2} = (|\mathsf{tmp1}{}_m|^2-|i_m|^2)\,\mathsf{tmp1}{}_m\bar{c}_n[\tau]/|\bar{c}_n[\tau]|$
\State $\mathsf{tmp3} = \mathsf{ifft}(\mathsf{tmp2})$
\State Compute $\partial L(\boldsymbol{c}[\tau])/\partial \bar{c}_n$ in \eqref{eq:grad}, where imaginary part is $\Im[\mathsf{tmp3}]$.
\State Update $c_n[\tau+1] \leftarrow c_n[\tau] - \alpha \partial L(\boldsymbol{c}[\tau])/\partial \bar{c}_n$.
\EndFor
\end{algorithmic}
\end{algorithm}
\subsection{Geometrical analysis}
\label{s3ss2}
In the WFPF, a phase-only hologram is parameterized by the phase only, 
as reviewed in Section~\ref{s2}.
During the update of $\phi_n$ in accordance with \eqref{eq:WFPF}, 
each element $A\exp(j\phi_n)$ evolves along a complex circle.
In contrast to the WFPF,
our hologram is parameterized by both the phase and auxiliary amplitude,
which is an intuitive difference between the previous gradient methods and ours.
We are now interested in identifying the key factors that enhance the quality of holograms in our method.
In this subsection,
we are dedicated to analyzing how the amplitude and phase of $c_n$, i.e.,
$A_n$ and $\phi_n$, are updated by repeating \eqref{eq:WFCF}.

To address this question, 
we first claim the following statement with respect to the phase.
\begin{theorem}
\label{t1}
The phase of the gradient satisfies
\begin{equation}
\mathop{\mathrm{arg}} \bigg(\frac{\partial}{\partial \bar{c}_n}L(\boldsymbol{c}[\tau])\bigg)=
\mathop{\mathrm{arg}} (c_n[\tau]) \pm \frac{\pi}{2}.
\label{eq:thm1}
\end{equation}
Equation~\eqref{eq:thm1} holds regardless of the initial value $c_n[0]$ and the distance $z$.
\end{theorem}
\begin{proof}
See \ref{a1}.
\end{proof}

\begin{remark}
\label{r1}
Figure~\ref{fig:flow_b} illustrates a geometrical interpretation of this result,
where the amplitude $|c_n[\tau]|$ and phase $\mathop{\mathrm{arg}}(c_n[\tau])$ are denoted by $A_n[\tau]$ and $\phi_n[\tau]$, respectively.
Theorem~\ref{t1} states that the gradient is a tangent vector at $A_n[\tau]\exp(j\phi_n[\tau])$ on a complex circle.
As shown in the left sides of Figs.~\ref{fig:flow_a} and \ref{fig:flow_b},
if the amplitudes of the gradients are sufficiently small, 
$\partial L/ \partial \bar{c}_n$ aligns with the linearly-approximated version of the directed arc $\partial L/\partial \phi_n$.
On the other hand, as shown in the right sides of Figs.~\ref{fig:flow_a} and \ref{fig:flow_b},
if the gradients have large amplitudes, $\partial L/ \partial \bar{c}_n$ indicates the different direction from that of  $\partial L/\partial \phi_n$.
Therefore, only if the amplitudes of the gradients are sufficiently small, 
the directions of both the gradients are almost parallel.
In other words, under the same condition,
the non-normalized hologram $A_n \exp(j\phi_n)$ evolves along a complex circle, 
just like the normalized hologram $A\exp(j\phi_n)$.
\end{remark}

Next, we show that the amplitude of the gradient is actually small under mild conditions.
We present the following two results with respect to the amplitude.
\begin{theorem}
\label{t2}
Suppose that there exists a complex vector $\boldsymbol{C} \in \mathbb{C}^N$ satisfying $|i_m|^2 = |p_m(\boldsymbol{C})|^2$, 
and that the distance between the hologram and screen planes is infinite, $z=\infty$.
Under these assumptions, the following inequality always holds:
\begin{equation}
\bigg|\frac{\partial}{\partial \bar{c}_n}L(\boldsymbol{c}[\tau])\bigg| 
\leq \frac{1}{|c_n[\tau]|}.
\label{eq:thm2}
\end{equation}
Equation~\eqref{eq:thm2}  holds regardless of the initial value $c_n[0]$.
\end{theorem}
\begin{proof}
See \ref{a2}.
\end{proof}
\begin{corollary}[induced from Theorem~\ref{t1}]
\label{c1}
The amplitude $|c_{n}[\tau]|$ monotonically increases:
\begin{equation}
|c_{n}[\tau+1]| \geq |c_{n}[\tau]|.
\label{eq:mono}
\end{equation}
The equality holds if and only if $\partial L(\boldsymbol{c}[\tau])/\partial \bar{c}_n=0$.
Equation~\eqref{eq:mono} holds regardless of the initial value $c_n[0]$ and the distance $z$.
\end{corollary}

\begin{proof}
We define $\nabla_n[\tau] = \alpha \partial L(\boldsymbol{c}[\tau])/\partial \bar{c}_n$ for notation convenience.
We compute the magnitude $|c_n[\tau+1]|^2$ as follows:
\begin{align}
|c_n[\tau+1]|^2 &= |c_n[\tau] - \nabla_n[\tau]|^2\\
&= |c_n[\tau]|^2 - \bar{c}_n[\tau]\nabla_n[\tau] 
- c_n[\tau]\bar{\nabla}_n[\tau] + |\nabla_n[\tau]|^2
\label{eq:proof_}\\
&= |c_n[\tau]|^2 + |\nabla_n[\tau]|^2
\label{eq:proof__}\\
&\geq |c_n[\tau]|^2.
\end{align}
We applied a variant of Theorem~\ref{t1}, i.e.,  $c_n[\tau] \bar{\nabla}_n[\tau] + \bar{c}_n[\tau] \nabla_n[\tau]=0$, to \eqref{eq:proof_}.
%
%
%
%
\end{proof}

\noindent
%
\begin{remark}
\label{r2}
Corollary~\ref{c1} claims that when $c_n[\tau]$ evolved enough, 
i.e., $\tau$ increased extremely, its amplitude $|A_n[\tau]|$ becomes sufficiently large.
With this characteristic in mind, we can interpret Theorem~\ref{t2} as follows: at a large $\tau$,
the reciprocal $1/|A_n[\tau]|$ can be sufficiently small, 
and the amplitude of the gradient is even less than $1/|A_n[\tau]|$.
\end{remark}
\begin{remark}
\label{r3}
In contrast to Remark~\ref{r2}, 
Theorem~\ref{t2} indicates that when $\tau$ is small, the amplitude of the gradient can be  large.
\end{remark}

\noindent
In accordance with Remark~\ref{r2}, 
the amplitude of the gradient can be sufficiently small at a large $\tau$;
each element $A_n[\tau] \exp(j\phi_n[\tau])$ thus evolves along a complex circle.
More precisely,
the amplitude $A_n[\tau]$ that was updated sufficient times gradually converges to its limit $A_n[\infty]$,
and each element approximately evolves along a circle of the limit radius $A_n[\infty]$,
as illustrated in the left side of Fig.~\ref{fig:flow_b}.
On the other hand,
Remark~\ref{r3} implies that $A_n[\tau]$ can drastically increase at a small $\tau$.
In other words,
each element $A_n[\tau] \exp(j\phi_n[\tau])$ is rapidly updated inside a circle of the limit radius $A_n[\infty]$,
as illustrated in the right side of Fig.~\ref{fig:flow_b}.
We conjecture that this update rule at a small $\tau$ accelerates the convergence towards the optimal hologram.

We finally present the supplementary remarks on Theorem~\ref{t2} and Corollary~\ref{c1}.

\begin{remark}
\label{r5}
The assumption $z=\infty$ in Theorem~\ref{t2} is intuitively impractical for an optical implementation of holography.
However, we can perceive the displayed image using Fourier relay lenses~\cite{Goodman05}.
\end{remark}

\begin{remark}
\label{r4}
Corollary~\ref{c1} states that for any $c_n[\tau]$,  since $|c_n[\tau]|$ is non-decreasing, 
$c_n[\tau]$ never degenerates to the zero point. 
\end{remark}

\section{Experimental results}
\label{s4}
\subsection{Comparisons with previous methods}
\label{s4ss1}
We conducted all the experiments on the computing environment in Table~\ref{tab:comp}.
The 8-bit gray-scale images in Figs.~\ref{fig:dpr_a} and \ref{fig:fmc_a} were used as planar images $|i_m|^2$.
These images were generated by cropping $1920 \times 1080$ regions of the original images in \cite{Bychkovsky11}.
The distance between the image and hologram planes was fixed to $z = \infty$, as assumed in Theorem~\ref{t2}.
We used the learning rate $\alpha=10^{-2}$ and $10^{-3}$.
The initial hologram $c_n[0]$ was drawn from a zero-mean complex Gaussian distribution with the variance 0.01.
\begin{table}[!t]
\caption{Computing environment.}
\vspace{-5mm}
\label{tab:comp}
\begin{center}
\footnotesize
\renewcommand{\arraystretch}{1.3}
\begin{tabular}{c|c}  \hline
GPU                     & NVIDIA GeForce RTX 4090\\ \hline
GPU memory   & 24~GB \\ \hline
OS          & Ubuntu 20.04 LTS \\ \hline
Language \& framework & Python 3.10.13 \& PyTorch 2.2.1\\\hline
\end{tabular}
\end{center}
\end{table}

We compared our method with three previous methods~\cite{Gerchberg72,Guendy21,Chakravarthula19},
which we call the GS method, Kaczmarz method, and WFPF, respectively.
These previous methods are known as the de facto standard for holographic phase retrieval, i.e.,
computation of phase-only holograms.
To verify the effectiveness of our method,
we evaluated the quality of displayed images using the peak signal-to-noise ratio (PSNR)~[dB].
For the GS and Kaczmarz methods,
we used $c_n[0]$ as initial estimates of phase-only holograms.
For the WFPF, we initialized $\phi_n[0]=\mathop{\mathrm{arg}}(c_n[0])$;
the initial holograms $A\exp(j\phi_n[0])$ and $A_n[0]\exp(j\phi_n[0])$ share the same phase,
but have different amplitudes.
We believe that this initialization scheme is suitable for investigating the effectiveness of the auxiliary amplitude $A_n$.
For the WFPF, we used the fixed learning rate $\alpha=10^{-2}$ and $10^{-3}$, as in our method.
We also implemented the WFPF using the Adam optimizer~\cite{Kingma15},
which adaptively changes the learning rate at each iteration $\tau$ and accelerates the gradient descent. 
We used the initial learning rate $10^{-2}$.
All the methods were executed by the GPU in Table~\ref{tab:comp}.

Figures~\ref{fig:dpr} and \ref{fig:fmc} show the displayed images and PSNR values in the simulation.
We used the number of iterations $\tau=10^3$ for all the methods.
The PSNR values ranged as follows: 
16.70--23.28 in the GS method,
15.82--20.04 in the Kaczmarz method,
17.23--20.17 in the WFPF with $\alpha = 10^{-2}$, and
19.47--25.00 in the WFPF with Adam.
Our method with $\alpha=10^{-2}$ presented PSNR values higher than those achieved by the previous methods,
which ranged from \textbf{19.90} to \textbf{25.65}.
The PSNR values of the WFPF and WFCF with $\alpha=10^{-3}$ were in the ranges 12.07--14.78 and 18.81--23.17,
respectively; the WFCF outperforms the WFPF despite employing the same learning rate.
In the colored regions in Figs.~\ref{fig:dpr} and \ref{fig:fmc},
the previous methods presented noisy results, 
which are obstacles for perceiving the textures and edges of objects.
On the other hand, noises were sufficiently suppressed in our results:
especially, facial lines (blue region) of \textit{Men}.

We optically displayed the \textit{Men} image by a holographic display.
Figure~\ref{fig:imp} shows the experimental system, 
where we used a Thorlabs EXULUS-HD1 phase-only SLM and a Point Grey GS3-U3-28S4C-C camera.
The parameters of our setup are listed in Table~\ref{tab:param}.
As shown in Fig.~\ref{fig:disp_a}, 
the WFPF with $\alpha=10^{-3}$ resulted in a noisy image, similar to the simulated one in Fig.~\ref{fig:fmc_e}.
In Fig.~\ref{fig:disp_b}, 
the WFPF with $\alpha=10^{-2}$ resulted in a less noisy but somewhat blurry image; 
see the face of the left man~(blue region).
Figure~\ref{fig:disp_c} shows our result with $\alpha=10^{-2}$,
where the face of the left man was clearer compared to those in Figs.~\ref{fig:disp_a} and \ref{fig:disp_b}.
These results support the effectiveness our method for real-world hardware implementation.

Figure~\ref{fig:psnr} shows PSNR values against the number of iterations $\tau$.
In Section~\ref{s3ss2},
we conjectured that the convergence of the non-normalized hologram is accelerated at a small $\tau$.
Figure~\ref{fig:psnr} shows that the WFCF with $\alpha=10^{-2}$ drastically increased PSNR values at a small $\tau$,
which demonstrates a higher-speed convergence compared with the WFPF.
Due to this characteristic, for most $\tau$ values, 
the WFCF presented PSNR values higher than those achieved by the previous methods.
Note that the WFPF using fixed learning rates gradually increased PSNR values.
This is because the WFPF evolves the normalized hologram along the unit complex circle, 
which leads to the slow convergence of the gradient descent.
These findings provide empirical supports for our conjecture regarding the acceleration.

The execution times for each iteration were as follows: 
2.64~[ms] for the GS method, 
3.08~[ms] for the Kaczmarz method,
2.68~[ms] for the WFPF with $\alpha=10^{-2}$,
2.80~[ms] for the WFPF with Adam, and
2.91~[ms] for the WFCF with $\alpha=10^{-2}$.
The differences in per-iteration complexity were negligible.
We also mention the execution times for achieving the PSNR value $25.0$~[dB] in Fig.~\ref{fig:psnr_a}:
48.77~[s] for the WFPF with $\alpha=10^{-2}$,
3.13~[s] for the WFPF with Adam, and
\textbf{1.80}~[s] for the WFCF with $\alpha=10^{-2}$.
We can observe that our method achieves the target PSNR with the smallest execution time, demonstrating its effectiveness from a computational complexity perspective.

\begin{remark}
It is notable that the PSNR values of the WFPF with Adam were comparable to those of our method.
In the field of machine learning, 
the rapid convergence of the Adam-based gradient method for certain classes of functions,
such as the non-convex loss in \eqref{eq:losspf}, is mysterious~\cite{Defossez22}, 
and its investigation is a challenging task.
We believe that our theoretical and experimental results can be helpful to address this challenge,
bridging the gap between holography and machine learning.
\end{remark}

\begin{remark}
The results presented in Figs.~\ref{fig:dpr} and \ref{fig:fmc} were degraded by salt-and-pepper noises,
which are probably due to the propagation distance $z$ and propagation model.
When $z \ll \infty$, i.e., in the case with the near-field propagation, 
the propagation equation in \eqref{eq:prop} can be regarded as Fresnel diffraction,
convolution of a Fresnel kernel and a hologram.
As was done in many works, e.g., \cite{Chakravarthula19},
the Fresnel diffraction is commonly implemented by the band-limited angular spectrum method, 
pioneered in \cite{Matsushima09}.
In the near-field propagation,
the band-limited angular spectrum method nullifies high-frequency components of the Fresnel kernel;
the Fresnel kernel can be regarded as a low-pass filter.
Therefore, only low-frequency components of holograms are propagated by the Fresnel diffraction,
whereas high-frequency noises, such as salt-and-pepper ones, 
can be removed, see \cite{Chakravarthula19} in the case $z=200$~[mm].
In contrast, when $z =\infty$, i.e., in the case with the far-field propagation,
the propagation equation is identical to the Fourier transformation.
In this case, all frequency components are uniformly propagated;
in other words, high-frequency components are propagated as much as low-frequency components.
Therefore, the results presented in Figs.~\ref{fig:dpr} and \ref{fig:fmc} reproduced sharp textures including facial lines of Men, as well as noises.
With the above discussion in mind,
we argue that obtaining high-PSNR values can be easy in the near-field propagation, as demonstrated in \cite{Chakravarthula19},
while it is challenging in the far-field propagation.
\end{remark}

\subsection{Numerical support for theoretical results}
\label{s4ss2}
To support the theoretical results in Section~\ref{s3ss2}, 
we evaluated the amplitudes and phases of $c_n[\tau]$ and $\partial L(\boldsymbol{c}[\tau])/\partial \bar{c}_n$.
We used the following synthetic image as $|i_m|^2$:
the non-normalized hologram $\boldsymbol{C}$ in Theorem~\ref{t2} was drawn from a zero-mean complex Gaussian distribution with the variance 0.04,
and $|i_m|^2$ was computed in accordance with the propagation equation $i_m=p_m(\boldsymbol{C})$.
This enables the synthetic image $|i_m|^2$ to satisfy the existence of  $\boldsymbol{C}$, 
which is a required condition for Theorem~\ref{t2}.
We also used the \textit{Rock} image in Fig.~\ref{fig:dpr_a} as $|i_m|^2$,
for which the existence of $\boldsymbol{C}$ is not guaranteed and remains uncertain.

Figure~\ref{fig:thorem_a} shows trajectories of $c_n[\tau]$ in the Gaussian plane,
where values beside the color bar represent the number of iterations $\tau$,
and the orange arrows represent the gradient $\partial L(\boldsymbol{c}[\tau])/\partial \bar{c}_n$.
We also illustrated complex circles of the limit radii $|c_n[\infty]|$  for each $n$, 
where we used $\tau=10^4$ as the infinity.
As can be seen, $c_n[\tau]$ was rapidly updated inside the corresponding complex circle at a small $\tau$,
whereas it evolved along a complex circles at a large $\tau$.
Figure~\ref{fig:thorem_b} shows the amplitude of the gradient, 
as well as its theoretical upper bound $1/|c_n[\tau]|$.
We can observe that $|\partial L(\boldsymbol{c}[\tau])/\partial \bar{c}_n|$ was significantly large at a small $\tau$. 
Figure~\ref{fig:thorem_c} shows the amplitudes of the interior points $c_n[\tau]$ in Fig.~\ref{fig:thorem_a},
where the red and blue lines represent $|c_n[\tau]|$ and its limit $|c_n[\infty]|$, respectively.
The amplitudes $|c_n[\tau]|$ drastically increased at a small $\tau$, and they converged to the limit,
as shown in Corollary~\ref{c1}.
These aforementioned results support the high-speed convergence of our method.
As can be seen from Fig.~\ref{fig:thorem_b},
the amplitude $|\partial L(\boldsymbol{c}[\tau])/\partial \bar{c}_n|$ was always lower than $1/|c_n[\tau]|$, as shown in Theorem~\ref{t2}.
Note that the difference between $1/|c_n[\tau]|$ and $|\partial L(\boldsymbol{c}[\tau])/\partial \bar{c}_n|$ was significant.
This is because we computed the bound $1/|c_n[\tau]|$ under the worst scenario, 
which was sufficiently large for our experimental setup.
Figure~\ref{fig:thorem_d} shows the phase $\mathrm{arg}(\partial L(\boldsymbol{c}[\tau])/\partial \bar{c}_n)-\mathrm{arg}(c_n[\tau])$ for a specific $n$.
Obviously, the phases are identical to $\pm\pi/2$, as shown in Theorem~\ref{t1}.
In accordance with all the results,
we conclude that even if the existence of $\boldsymbol{C}$ is uncertain,
our method exhibits consistent and theoretical behavior.

\section{Conclusion}
\label{s5}
In this work, we addressed a holographic phase retrieval problem:
optimization of a phase-only hologram while keeping the amplitude constant. 
We extended a phase-only hologram by introducing the auxiliary amplitude.
We also defined a gradient method for optimizing the phase and amplitude,
which we called the Wirtinger flow in Cartesian form~(WFCF).
At the early stage of optimization, 
each element of the hologram exists inside a complex circle,
and it can take a large gradient while diverging from the origin.
This characteristic contributed to accelerating the gradient descent.
Comparing the WFCF with previous methods, 
we demonstrated its effectiveness in simulation and optical reproduction.
%
%
We hope our methodology contributes to further advancements in the field of holography.




\bibliography{sample}

\newpage
\begin{figure}[!t]
\centering
\begin{subfigure}{0.996\textwidth}
\centering
\includegraphics[width=0.29\textwidth]{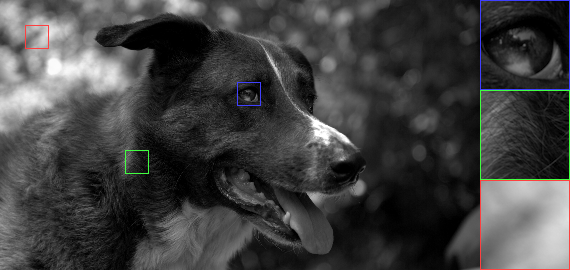}
\hfill
\includegraphics[width=0.29\textwidth]{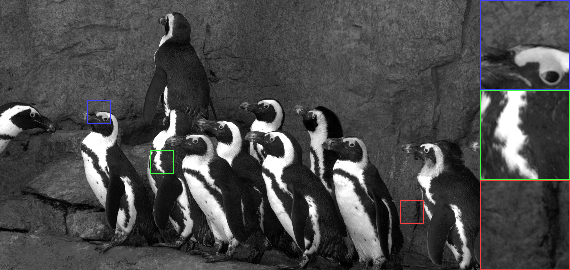}
\hfill
\includegraphics[width=0.29\textwidth]{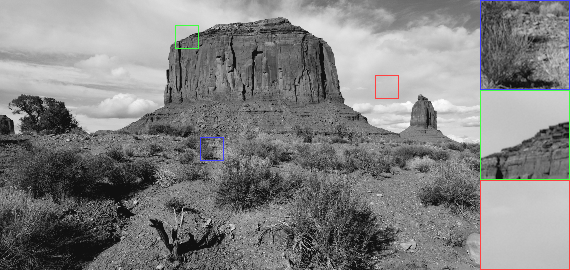}
\caption{Original images~(\textit{Dog}, \textit{Penguins}, and \textit{Rock})}
\label{fig:dpr_a}
\vspace{1.7mm}
\end{subfigure}
\begin{subfigure}{0.996\textwidth}
\centering
\includegraphics[width=0.29\textwidth]{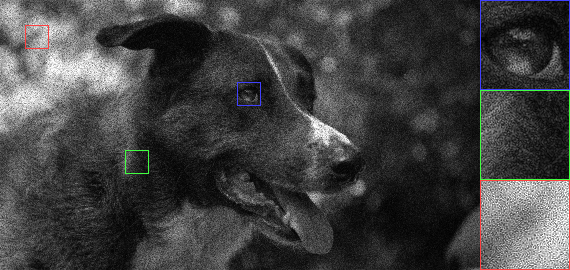}
\hfill
\includegraphics[width=0.29\textwidth]{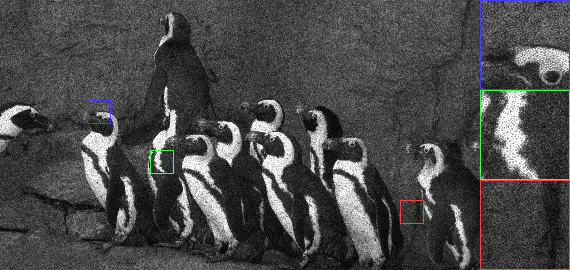}
\hfill
\includegraphics[width=0.29\textwidth]{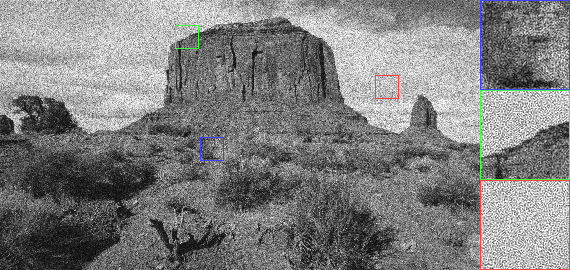}
\caption{Gerchberg-Saxton method~\cite{Gerchberg72}~(22.59, 21.40, and 17.53)}
\vspace{1.7mm}
\end{subfigure}
\begin{subfigure}{0.996\textwidth}
\centering
\includegraphics[width=0.29\textwidth]{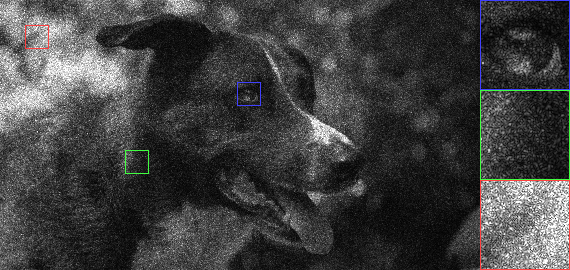}
\hfill
\includegraphics[width=0.29\textwidth]{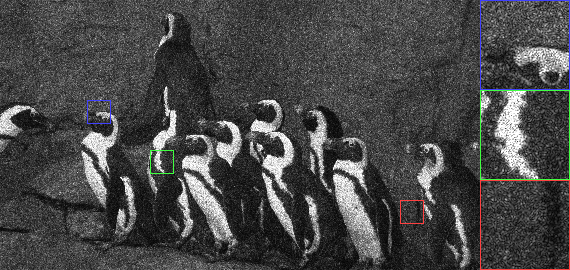}
\hfill
\includegraphics[width=0.29\textwidth]{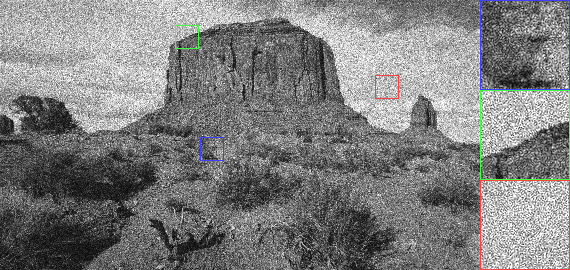}
\caption{Kaczmarz method~\cite{Guendy21}~(19.85, 19.31, and 16.47)}
\vspace{1.7mm}
\end{subfigure}
\begin{subfigure}{0.996\textwidth}
\centering
\includegraphics[width=0.29\textwidth]{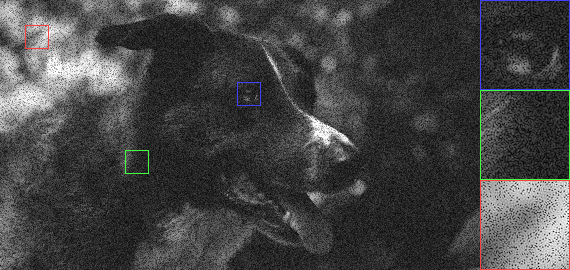}
\hfill
\includegraphics[width=0.29\textwidth]{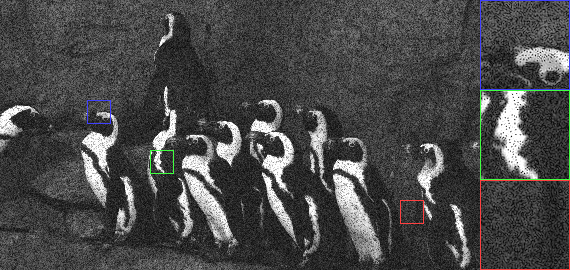}
\hfill
\includegraphics[width=0.29\textwidth]{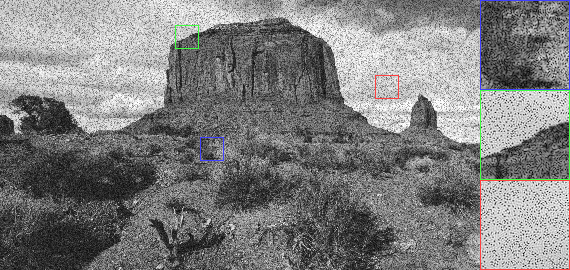}
\caption{WFPF~\cite{Chakravarthula19} w/ $\alpha=10^{-2}$~(20.17, 20.16, and 17.75)}
\vspace{1.7mm}
\end{subfigure}
\begin{subfigure}{0.996\textwidth}
\centering
\includegraphics[width=0.29\textwidth]{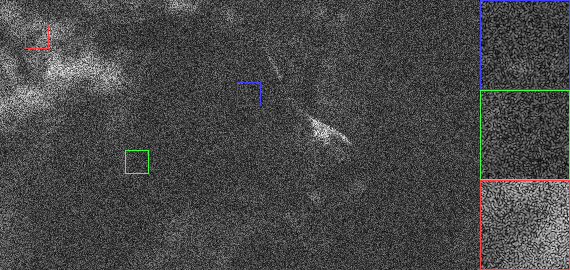}
\hfill
\includegraphics[width=0.29\textwidth]{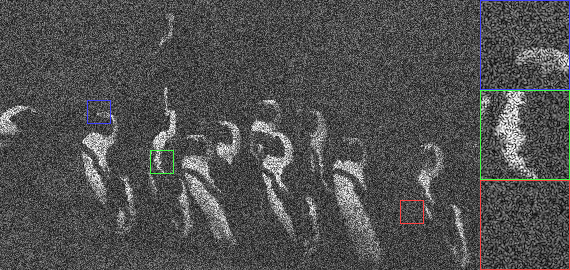}
\hfill
\includegraphics[width=0.29\textwidth]{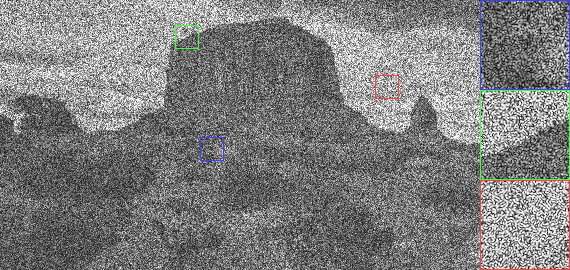}
\caption{WFPF~\cite{Chakravarthula19} w/ $\alpha=10^{-3}$~(14.76, 14.78, and 12.68)}
\vspace{1.7mm}
\end{subfigure}
\begin{subfigure}{0.996\textwidth}
\centering
\includegraphics[width=0.29\textwidth]{img0_PF_2.png}
\hfill
\includegraphics[width=0.29\textwidth]{img1_PF_2.png}
\hfill
\includegraphics[width=0.29\textwidth]{img2_PF_2.png}
\caption{WFPF~\cite{Chakravarthula19} w/ Adam~(24.83, 24.08, and 20.80)}
\vspace{1.7mm}
\end{subfigure}
\begin{subfigure}{0.996\textwidth}
\centering
\includegraphics[width=0.29\textwidth]{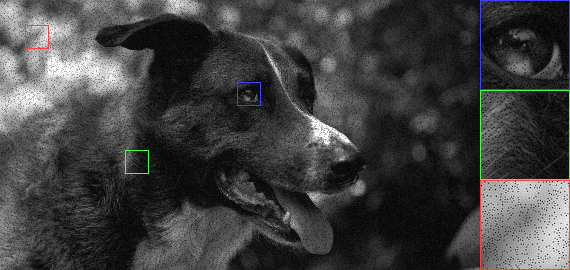}
\hfill
\includegraphics[width=0.29\textwidth]{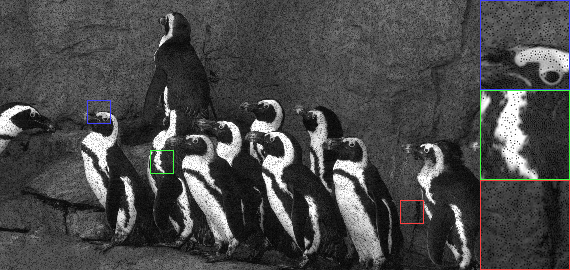}
\hfill
\includegraphics[width=0.29\textwidth]{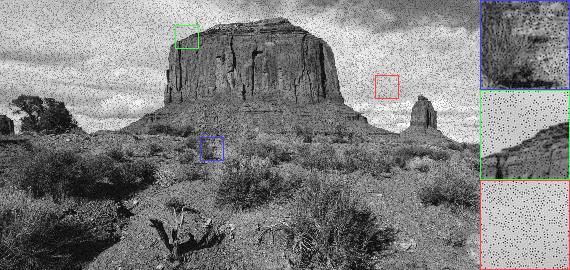}
\caption{WFCF~(ours) w/ $\alpha=10^{-2}$~(\textbf{25.33}, \textbf{24.59}, and \textbf{21.09})}
\vspace{1.7mm}
\end{subfigure}
\begin{subfigure}{0.996\textwidth}
\centering
\includegraphics[width=0.29\textwidth]{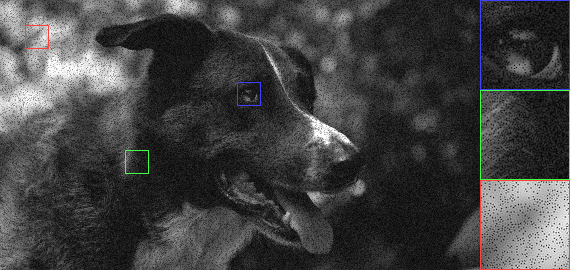}
\hfill
\includegraphics[width=0.29\textwidth]{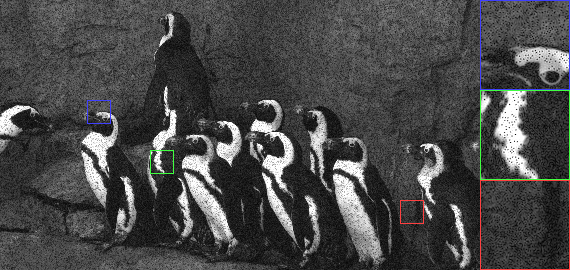}
\hfill
\includegraphics[width=0.29\textwidth]{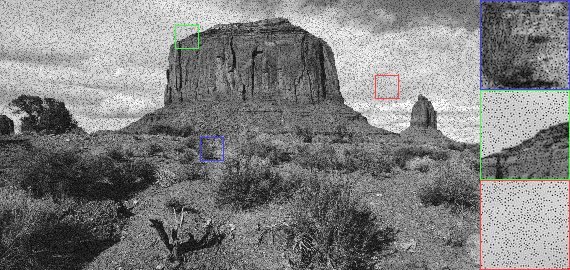}
\caption{WFCF~(ours) w/ $\alpha=10^{-3}$~(23.17, 22.71, and 19.86)}
\end{subfigure}
\caption{Displayed images and PSNR values in simulation~(\textit{Dog}, \textit{Penguins}, and \textit{Rock}).}
\label{fig:dpr}
\end{figure}

\newpage
\begin{figure}[!t]
\centering
\begin{subfigure}{0.996\textwidth}
\centering
\includegraphics[width=0.29\textwidth]{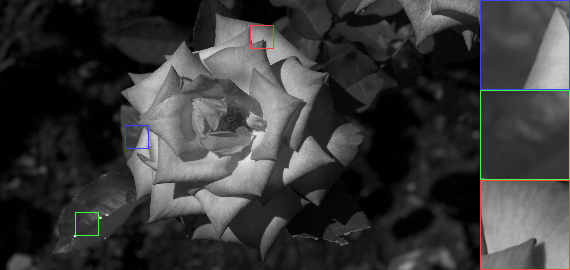}
\hfill
\includegraphics[width=0.29\textwidth]{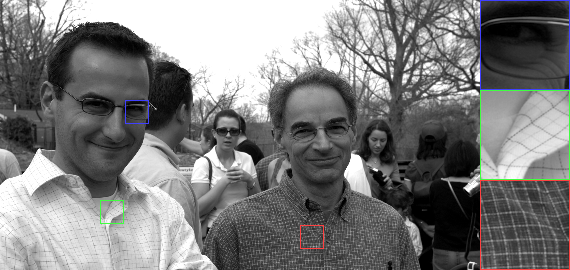}
\hfill
\includegraphics[width=0.29\textwidth]{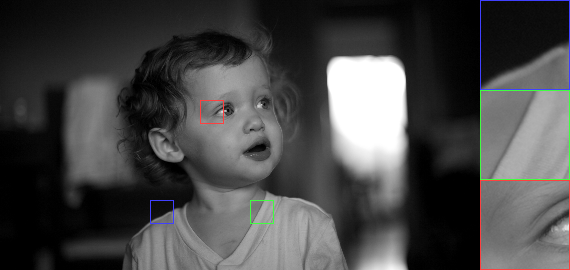}
\caption{Original images~(\textit{Flower}, \textit{Men}, and  \textit{Baby})}
\label{fig:fmc_a}
\vspace{1.7mm}
\end{subfigure}
\begin{subfigure}{0.996\textwidth}
\centering
\includegraphics[width=0.29\textwidth]{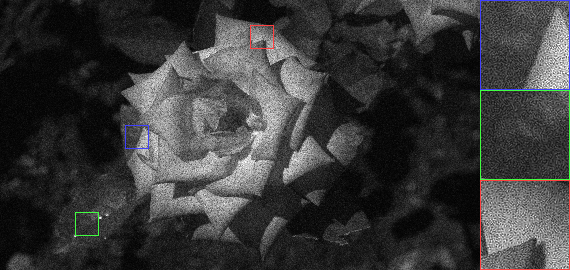}
\hfill
\includegraphics[width=0.29\textwidth]{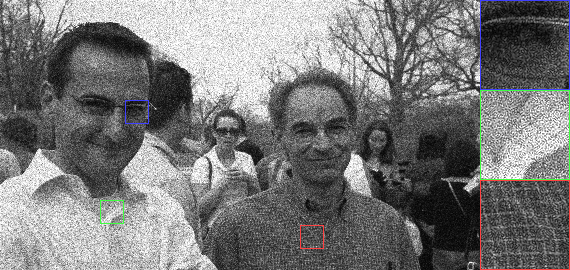}
\hfill
\includegraphics[width=0.29\textwidth]{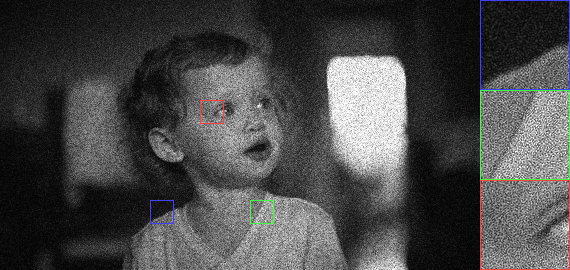}
\caption{Gerchberg-Saxton method~\cite{Gerchberg72}~(23.28, 16.70, and 20.74)}
\vspace{1.7mm}
\end{subfigure}
\begin{subfigure}{0.996\textwidth}
\centering
\includegraphics[width=0.29\textwidth]{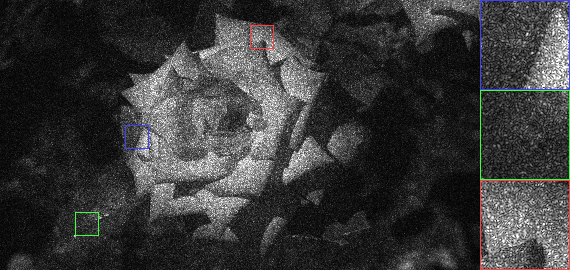}
\hfill
\includegraphics[width=0.29\textwidth]{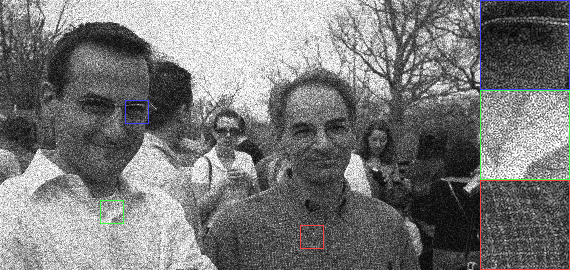}
\hfill
\includegraphics[width=0.29\textwidth]{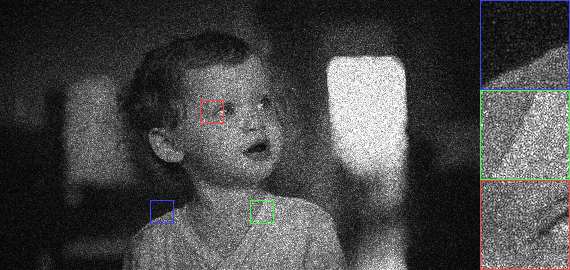}
\caption{Kaczmarz method~\cite{Guendy21}~(20.04, 15.82, and 18.80)}
\vspace{1.7mm}
\end{subfigure}
\begin{subfigure}{0.996\textwidth}
\centering
\includegraphics[width=0.29\textwidth]{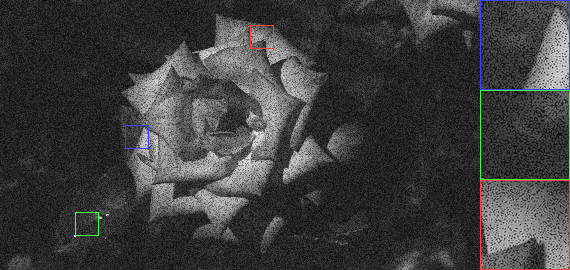}
\hfill
\includegraphics[width=0.29\textwidth]{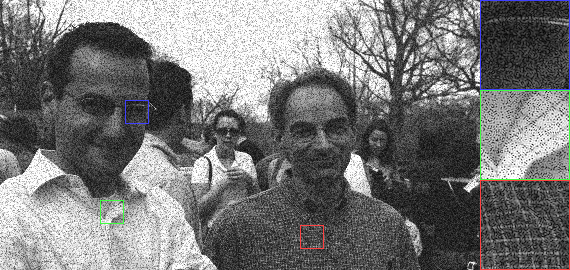}
\hfill
\includegraphics[width=0.29\textwidth]{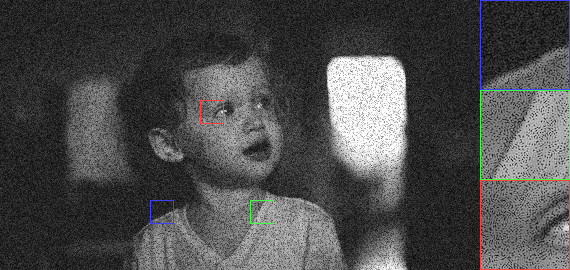}
\caption{WFPF~\cite{Chakravarthula19} w/ $\alpha=10^{-2}$~(19.50, 17.23, and 18.79)}
\vspace{1.7mm}
\end{subfigure}
\begin{subfigure}{0.996\textwidth}
\centering
\includegraphics[width=0.29\textwidth]{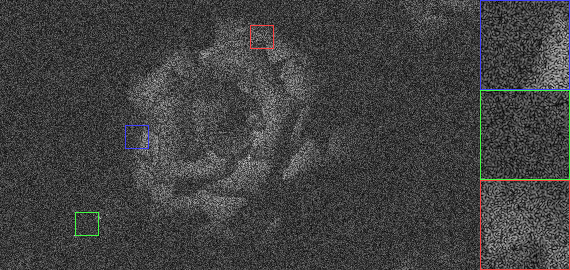}
\hfill
\includegraphics[width=0.29\textwidth]{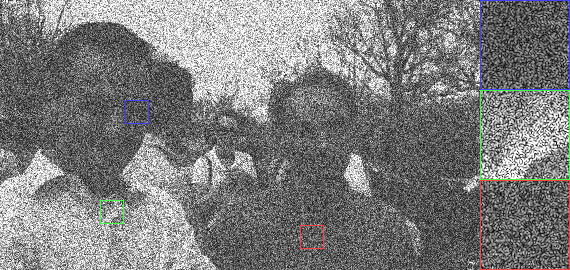}
\hfill
\includegraphics[width=0.29\textwidth]{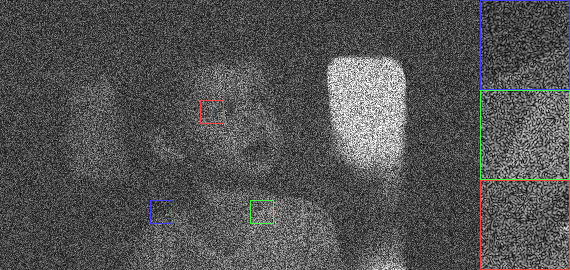}
\caption{WFPF~\cite{Chakravarthula19} w/ $\alpha=10^{-3}$~(14.12, 12.07, and 12.87)}
\label{fig:fmc_e}
\vspace{1.7mm}
\end{subfigure}
\begin{subfigure}{0.996\textwidth}
\centering
\includegraphics[width=0.29\textwidth]{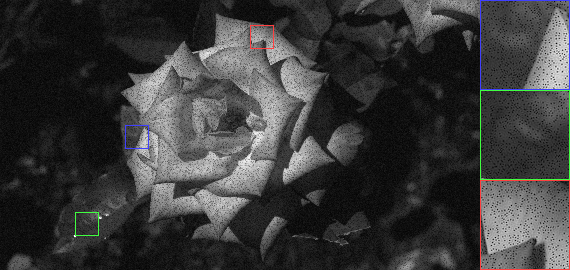}
\hfill
\includegraphics[width=0.29\textwidth]{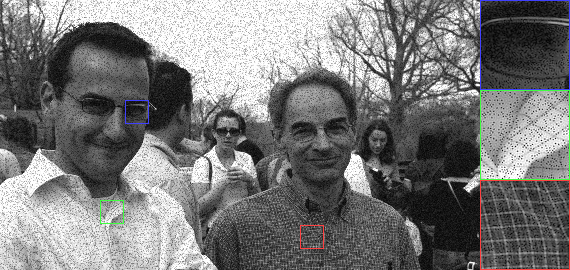}
\hfill
\includegraphics[width=0.29\textwidth]{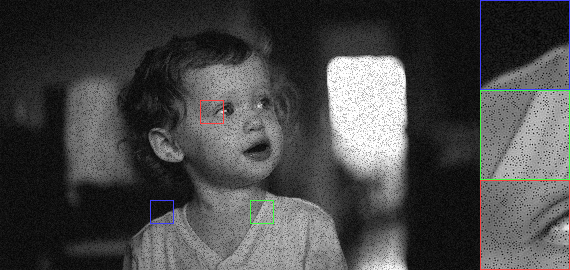}
\caption{WFPF~\cite{Chakravarthula19} w/ Adam~(25.00, 19.47, and 22.97)}
\vspace{1.7mm}
\end{subfigure}
\begin{subfigure}{0.996\textwidth}
\centering
\includegraphics[width=0.29\textwidth]{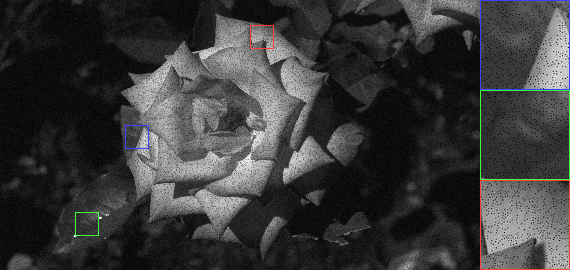}
\hfill
\includegraphics[width=0.29\textwidth]{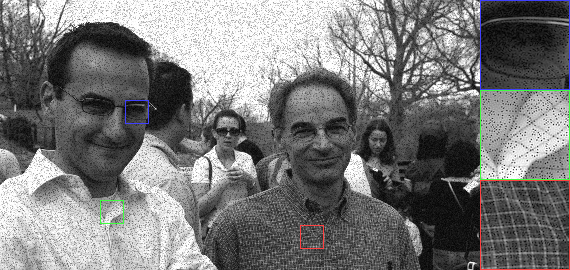}
\hfill
\includegraphics[width=0.29\textwidth]{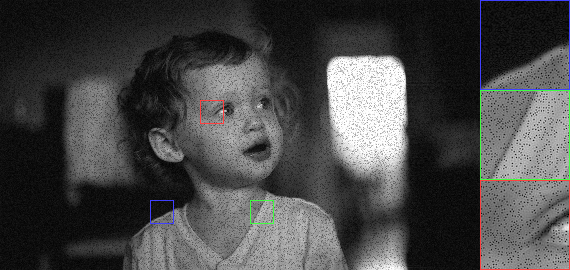}
\caption{WFCF~(ours) w/ $\alpha=10^{-2}$~(\textbf{25.65}, \textbf{19.90}, and \textbf{23.49})}
\vspace{1.7mm}
\end{subfigure}
\begin{subfigure}{0.996\textwidth}
\centering
\includegraphics[width=0.29\textwidth]{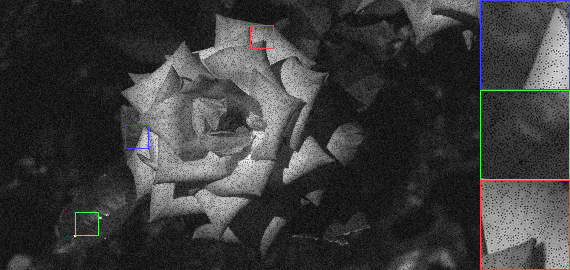}
\hfill
\includegraphics[width=0.29\textwidth]{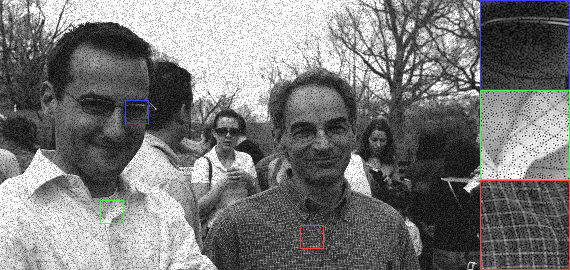}
\hfill
\includegraphics[width=0.29\textwidth]{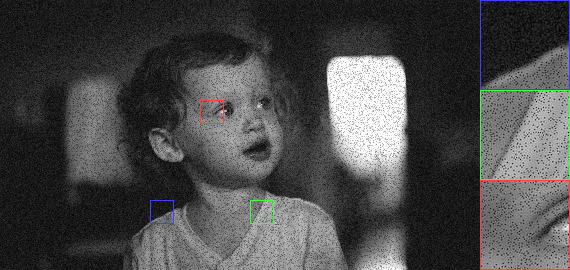}
\caption{WFCF~(ours) w/ $\alpha=10^{-3}$~(22.79, 18.81, and 21.49)}
\end{subfigure}
\caption{Displayed images and PSNR values in simulation~(\textit{Flower}, \textit{Men}, and \textit{Baby}).}
\label{fig:fmc}
\end{figure}

\clearpage
\begin{figure}[!t]
\centering
\includegraphics[scale=1.09]{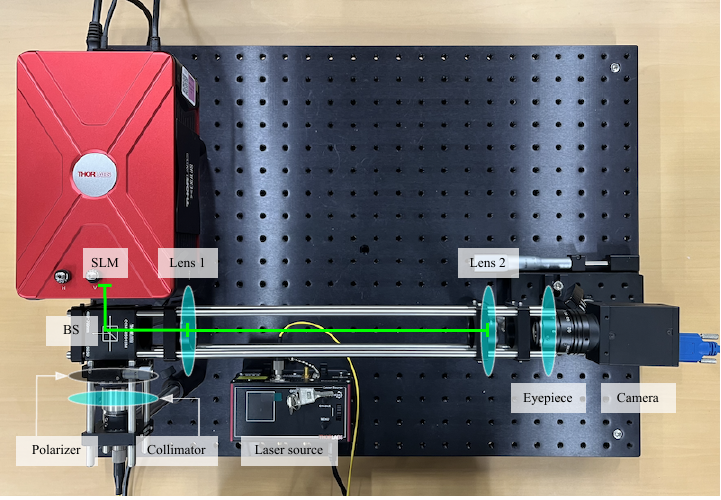}\hspace{-0.50mm}
\caption{Optical implementation of holography.}
\label{fig:imp}
\vspace{5mm}
\end{figure}
\begin{table}[!t]
\caption{Parameters of optical system in Fig.~\ref{fig:imp}.}
\label{tab:param}
\vspace{-5mm}
\begin{center}
\footnotesize
\renewcommand{\arraystretch}{1.3}
\begin{tabular}{wl{50mm}|wr{30mm}}  \hline
Distance between SLM and lens~1 & 82~[mm] \\ \hline
Distance between lenses~1 and 2 & 220~[mm]\\ \hline
Focal length of lens~1      & 150~[mm] \\ \hline
Focal length of lens~2      & 100~[mm] \\ \hline
Focal length of eyepiece & 35~[mm] \\ \hline
Wavelength of laser & 635~[nm] \\ \hline
Pixel pitch of SLM   & 6.4~[\textmu m] \\ \hline
Pixel pitch of camera        & 3.69~[\textmu m] \\ \hline
Resolution of SLM & 1920 $\times$ 1080~[pixels]\\\hline
Resolution of camera & 1928 $\times$ 1448~[pixels]\\\hline
\end{tabular}
\end{center}
\vspace{5mm}
\end{table}
\begin{figure}[!t]
\begin{subfigure}{0.33\textwidth}
\centering
\includegraphics[width=0.98\textwidth]{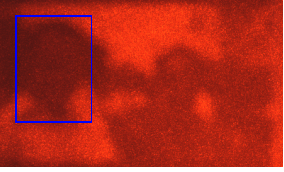}
\caption{WFPF~\cite{Chakravarthula19} w/ $\alpha=10^{-3}$}
\label{fig:disp_a}
\vspace{1.5mm}
\end{subfigure}\hspace{0.1mm}
\begin{subfigure}{0.33\textwidth}
\centering
\includegraphics[width=0.98\textwidth]{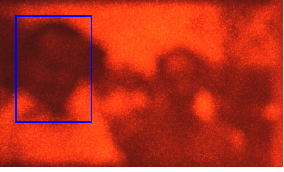}
\caption{WFPF~\cite{Chakravarthula19} w/ $\alpha=10^{-2}$}
\label{fig:disp_b}
\vspace{1.5mm}
\end{subfigure}
\begin{subfigure}{0.33\textwidth}
\centering
\includegraphics[width=0.98\textwidth]{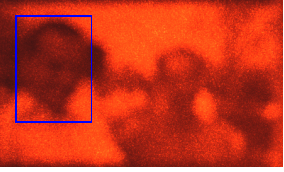}
\caption{WFCF~(ours) w/ $\alpha=10^{-2}$}
\label{fig:disp_c}
\vspace{1.5mm}
\end{subfigure}\hspace{-0.3mm}
\caption{Optically displayed images~(\textit{Men}).}
\label{fig:disp}
\end{figure}

\newpage
\begin{figure}[!t]
\centering
\includegraphics[width=0.84\textwidth]{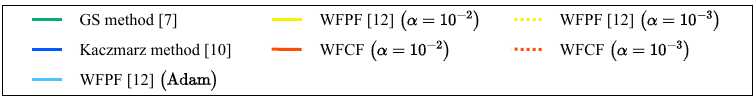}\\
\vspace{4mm}
\begin{subfigure}{0.49\textwidth}
\centering
\includegraphics[width=0.99\textwidth]{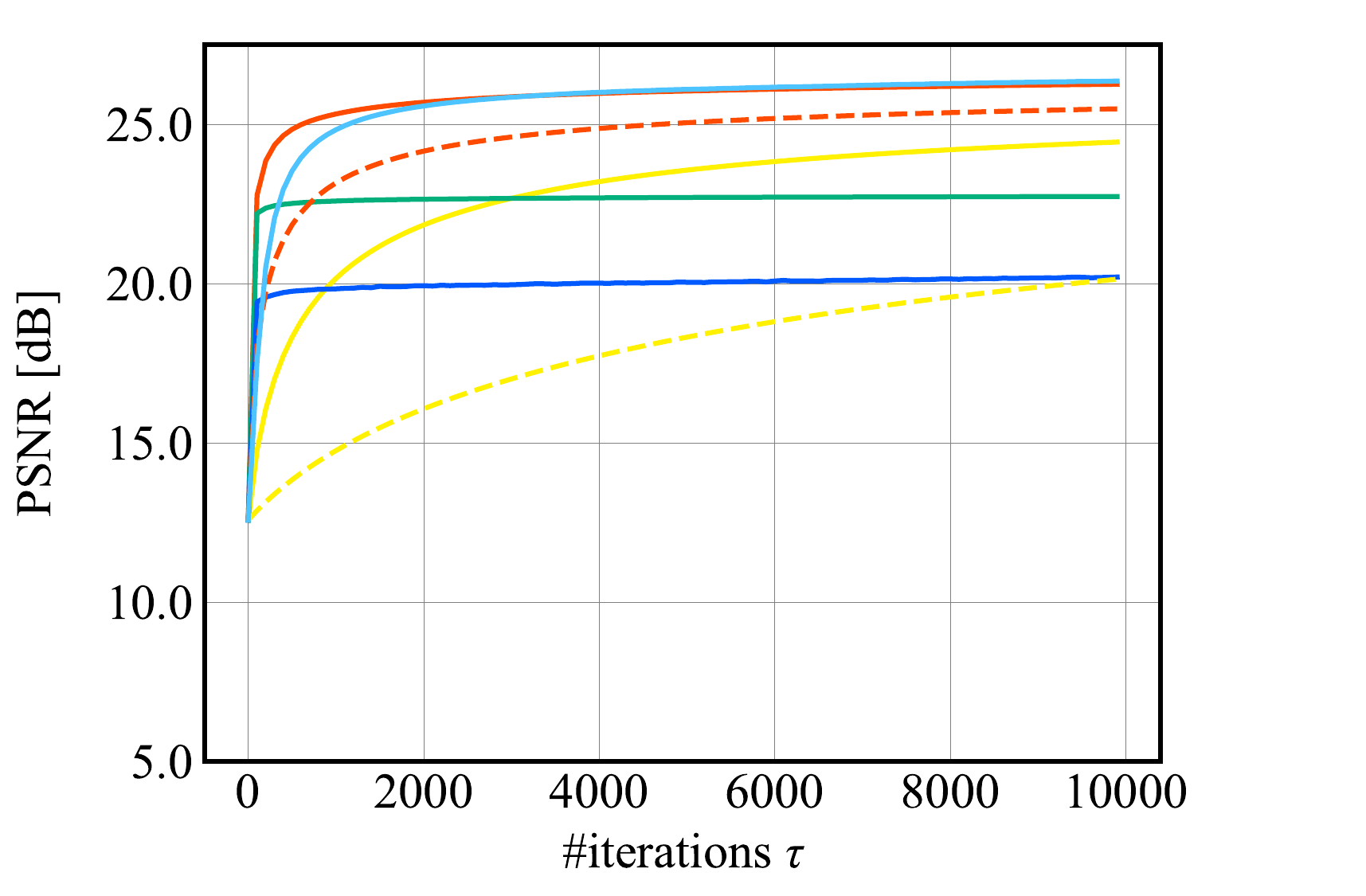}
\caption{\textit{Dog}}
\label{fig:psnr_a}
\vspace{2.0mm}
\end{subfigure}
\begin{subfigure}{0.49\textwidth}
\centering
\includegraphics[width=0.99\textwidth]{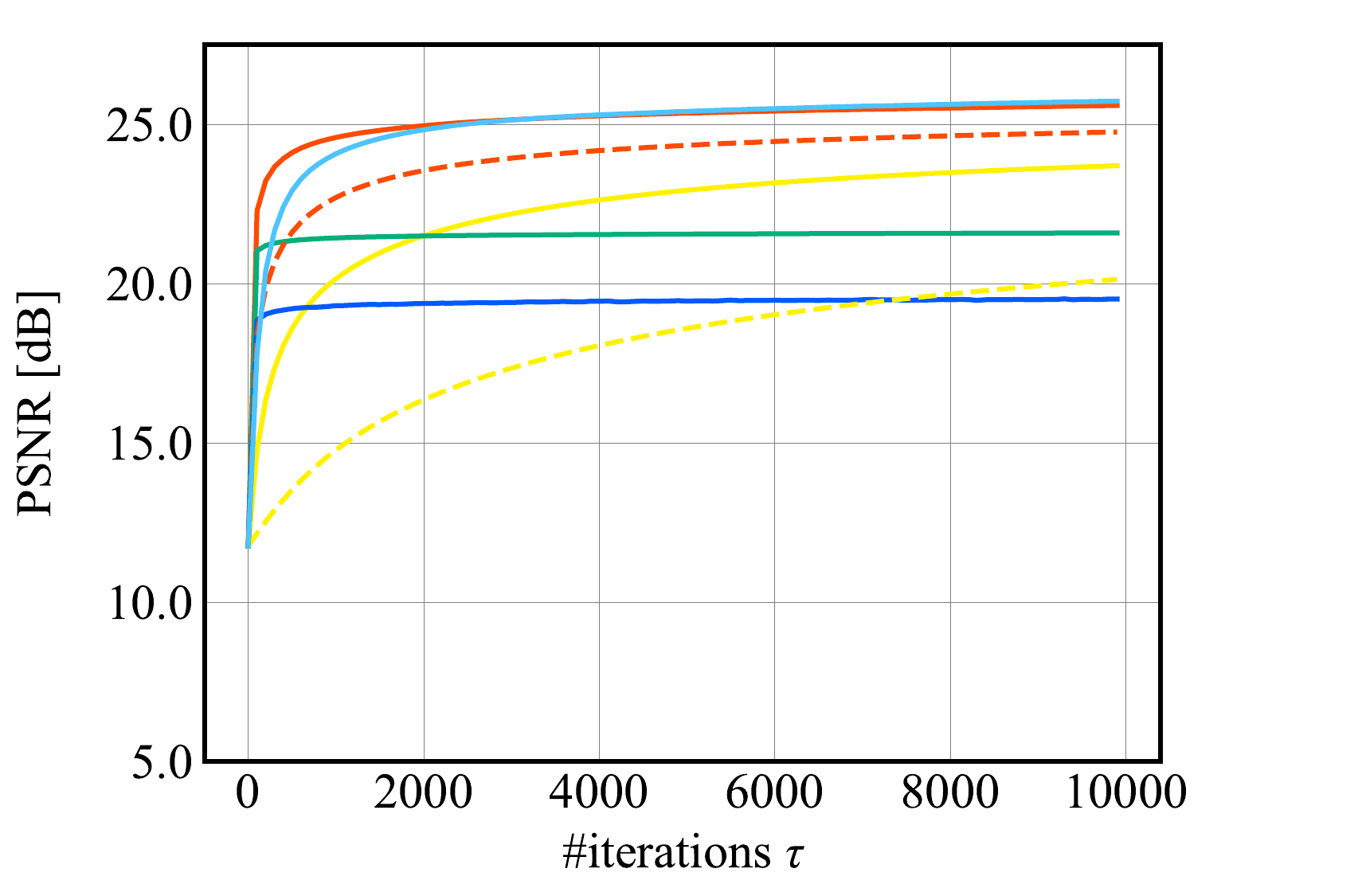}
\caption{\textit{Penguins}}
\vspace{2.0mm}
\end{subfigure}\\
\begin{subfigure}{0.49\textwidth}
\centering
\includegraphics[width=0.99\textwidth]{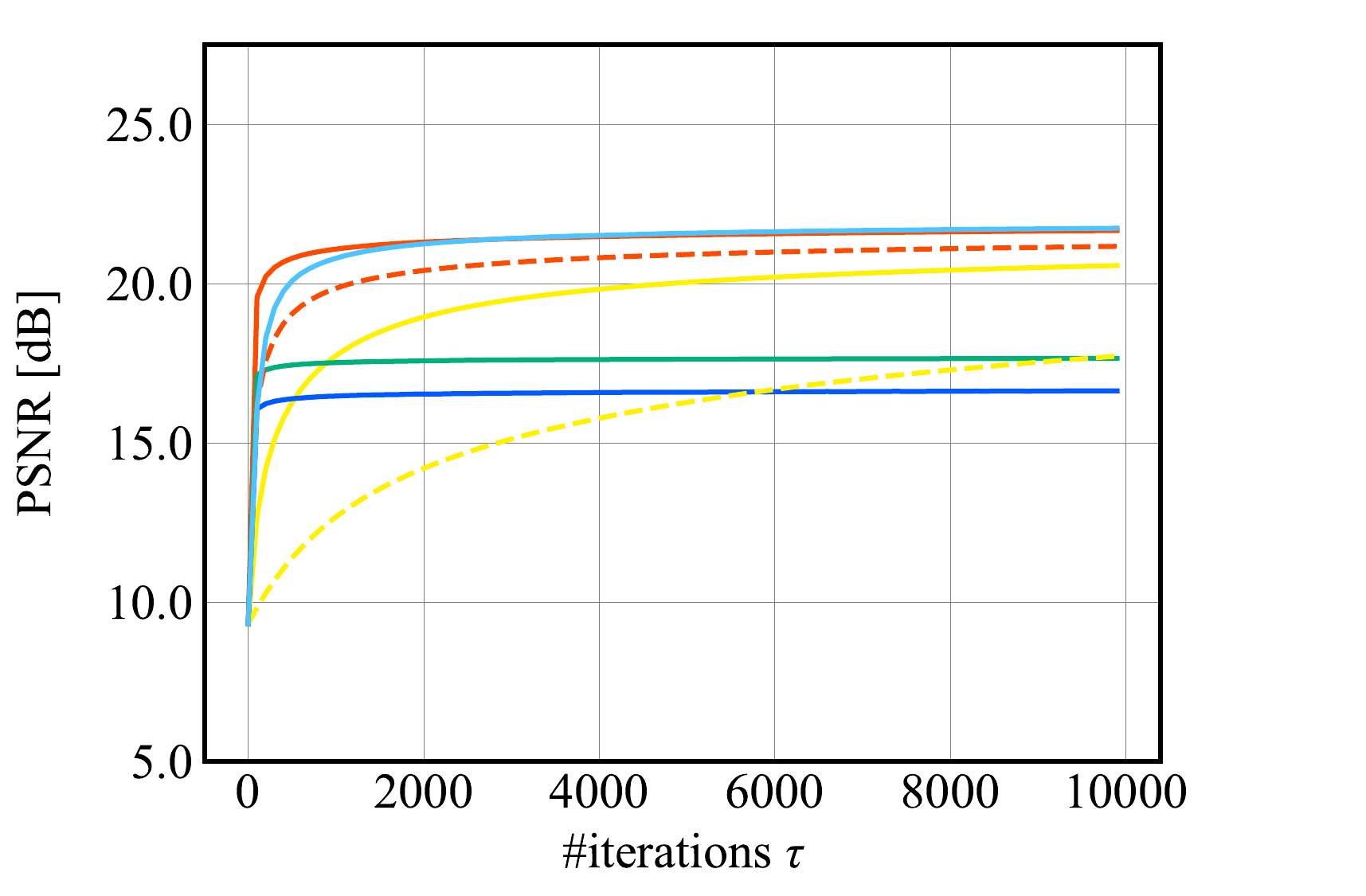}
\caption{\textit{Rock}}
\vspace{2.0mm}
\end{subfigure}
\begin{subfigure}{0.49\textwidth}
\centering
\includegraphics[width=0.99\textwidth]{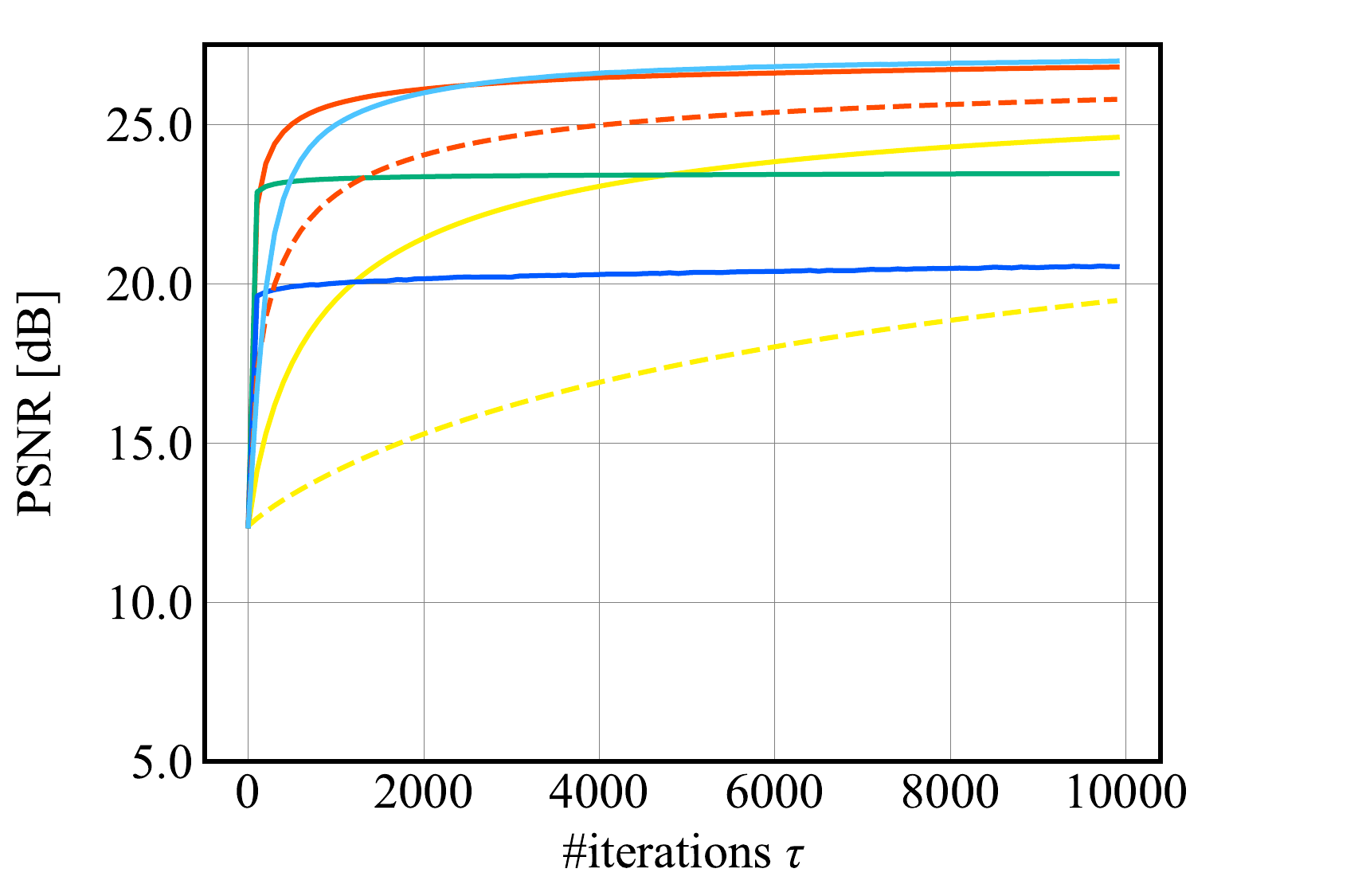}
\caption{\textit{Flower}}
\vspace{2.0mm}
\end{subfigure}\\
\begin{subfigure}{0.49\textwidth}
\centering
\includegraphics[width=0.99\textwidth]{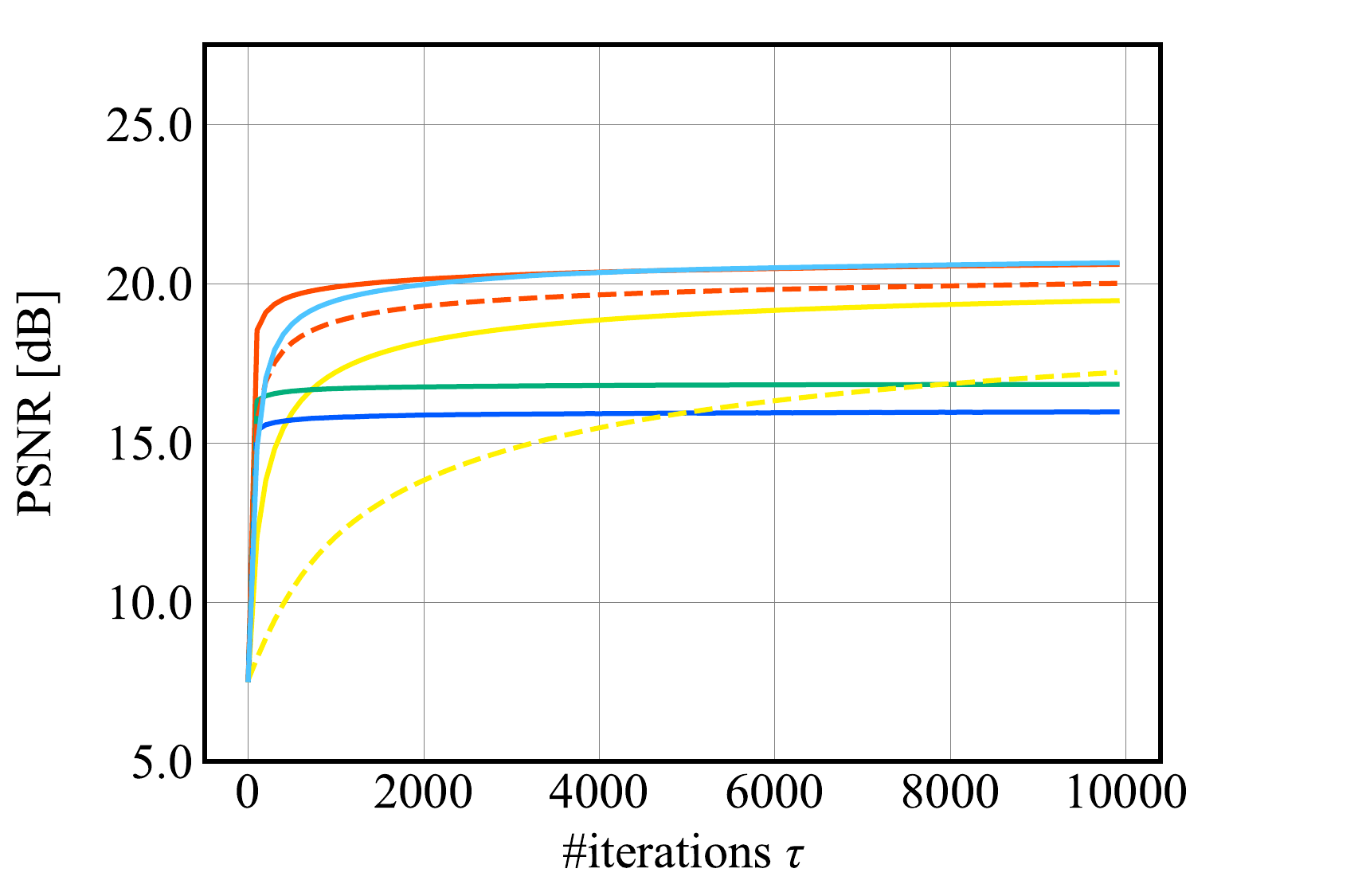}
\caption{\textit{Men}}
\vspace{2.0mm}
\end{subfigure}
\begin{subfigure}{0.49\textwidth}
\centering
\includegraphics[width=0.99\textwidth]{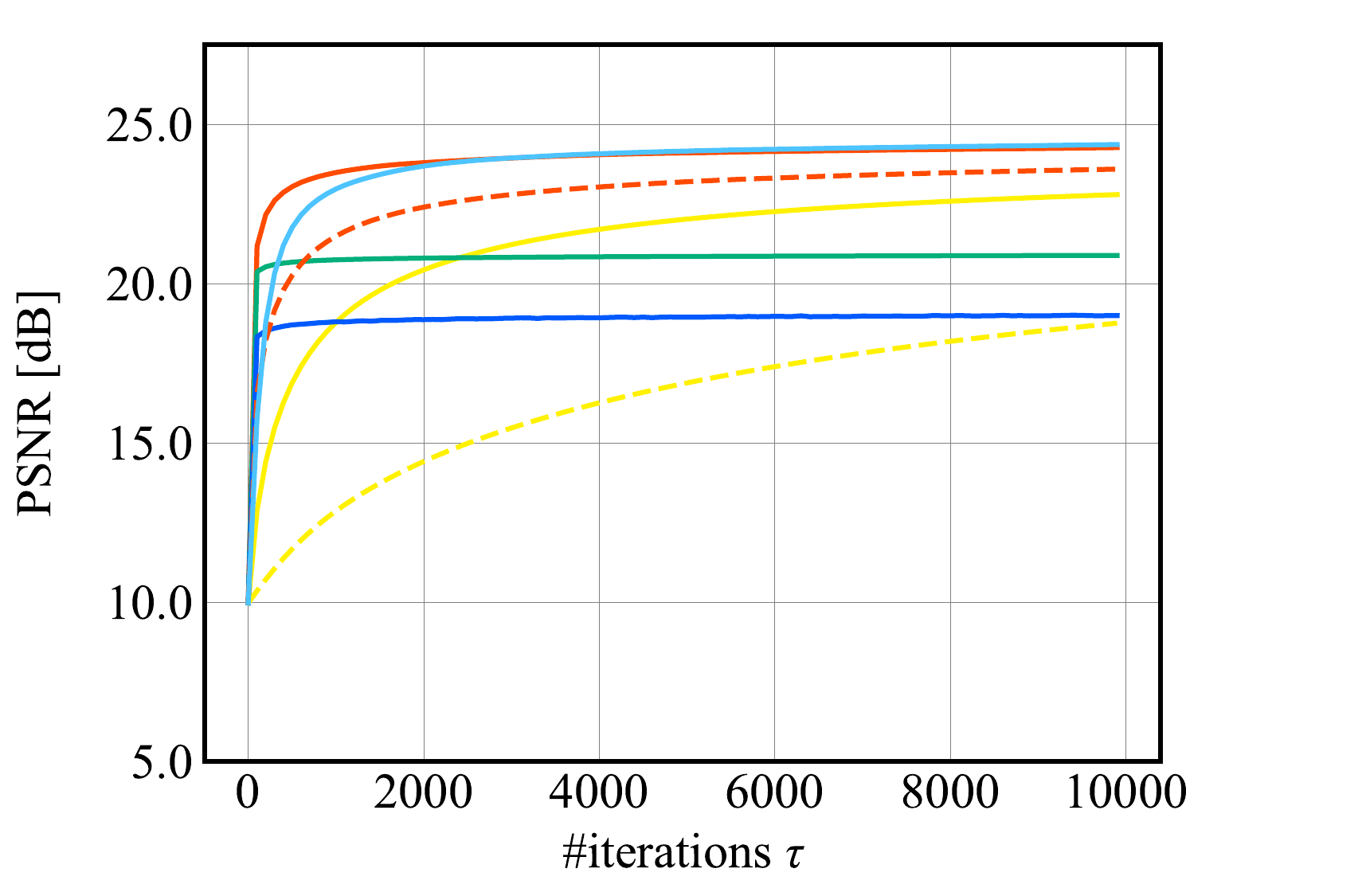}
\caption{\textit{Baby}}
\vspace{2.0mm}
\end{subfigure}
\caption{PSNR values against number of iterations.}
\label{fig:psnr}
\end{figure}

\newpage
\begin{figure}[!t]
\centering
\begin{subfigure}{0.99\textwidth}
\centering
\includegraphics[width=0.49\textwidth]{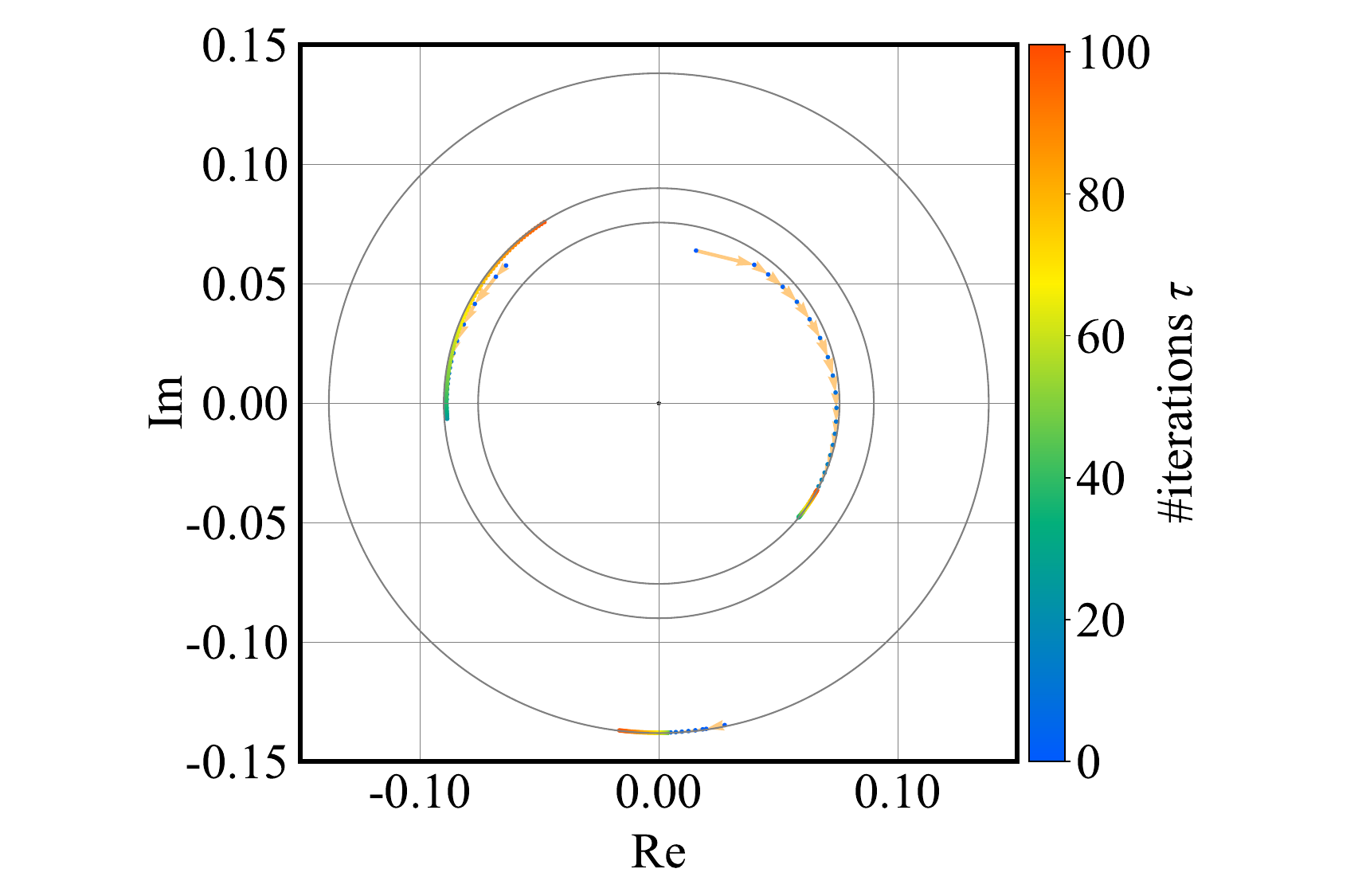}\hspace{0.6mm}
\includegraphics[width=0.49\textwidth]{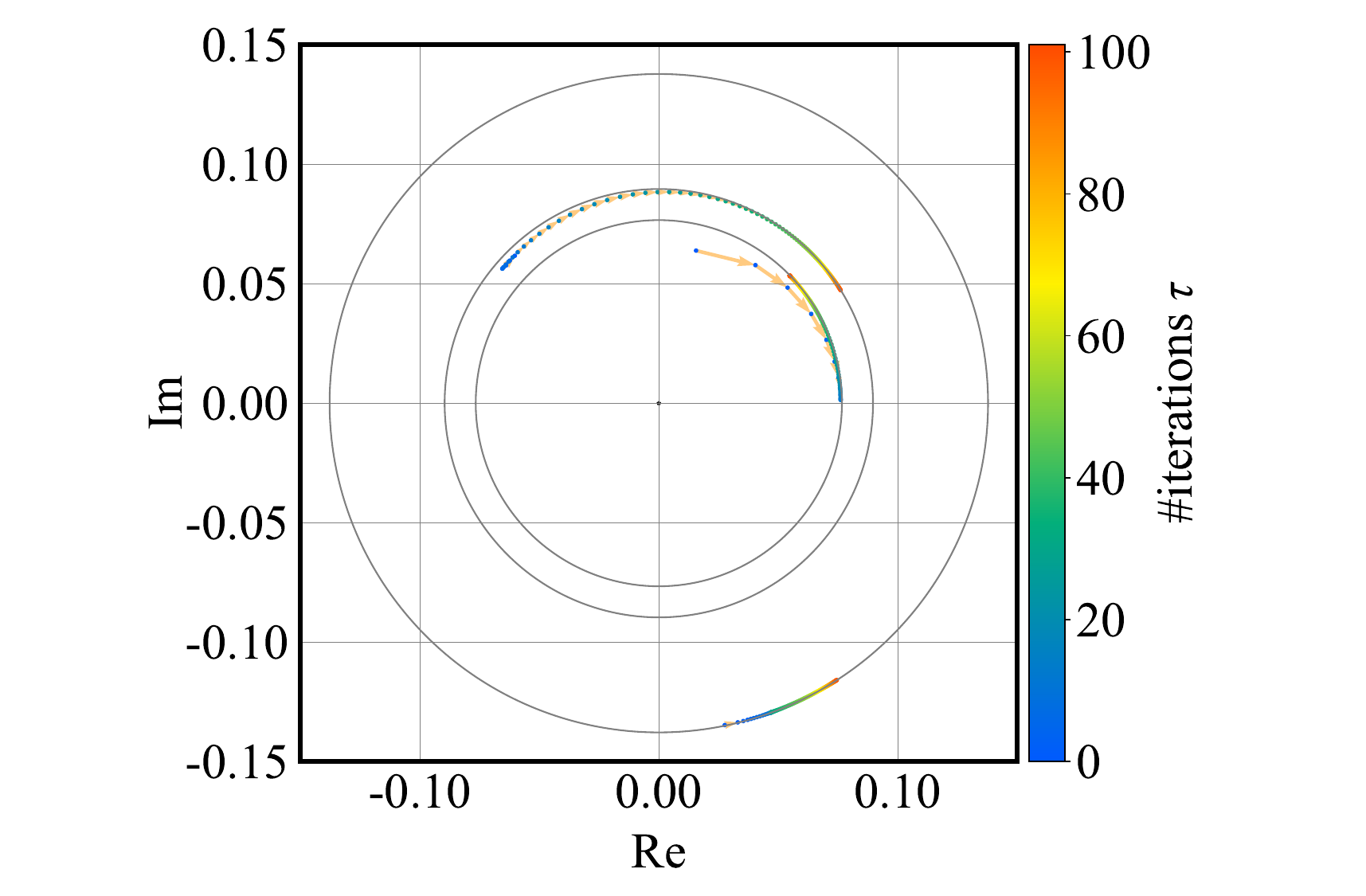}\hspace{-2.4mm}
\caption{$c_n[\tau]$ and circles of limit radii $|c_n[\infty]|$}
\label{fig:thorem_a}
\vspace{3.5mm}
\end{subfigure}\\
\begin{subfigure}{0.99\textwidth}
\centering
\includegraphics[width=0.49\textwidth]{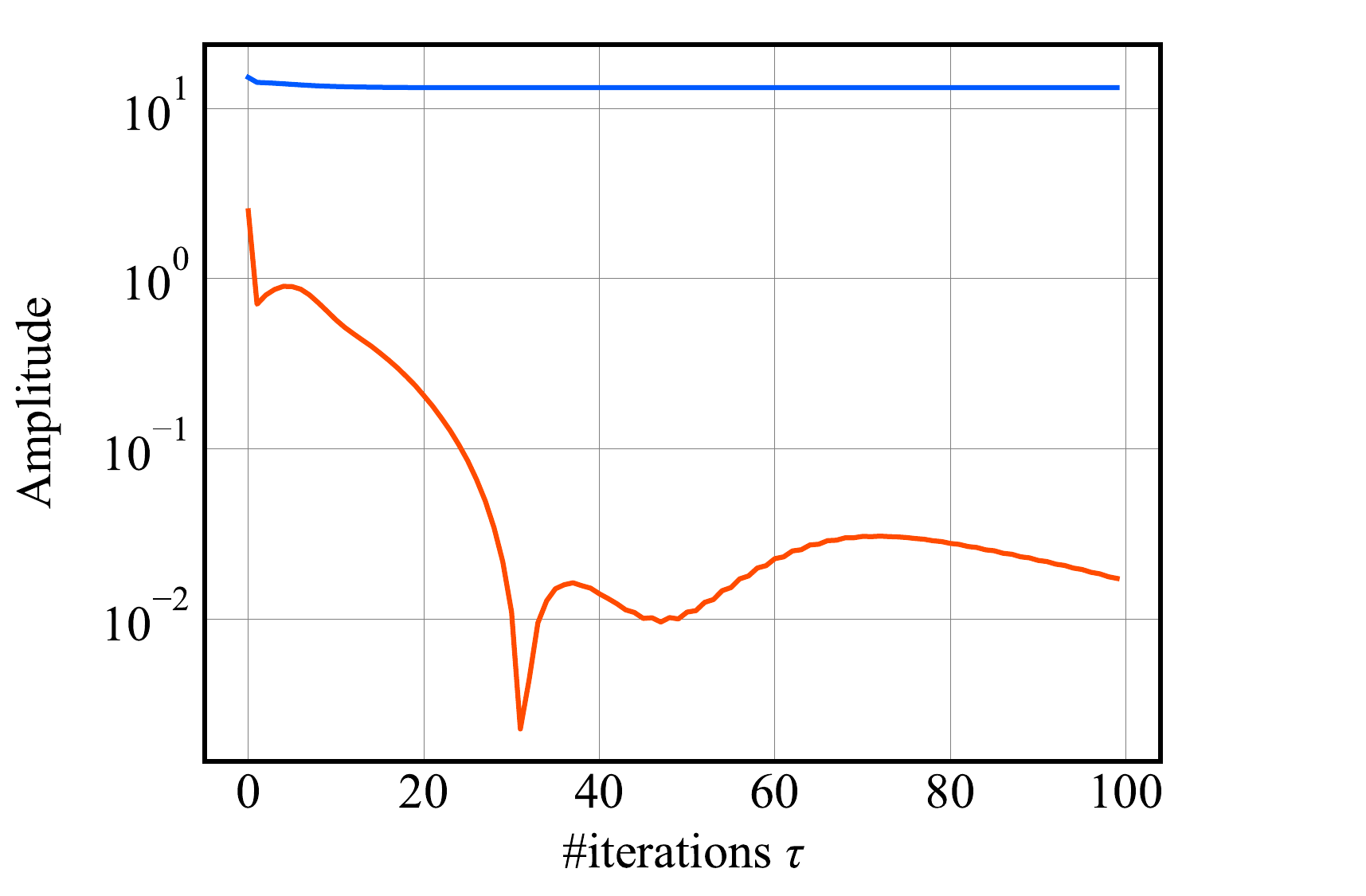}\hspace{0.7mm}
\includegraphics[width=0.49\textwidth]{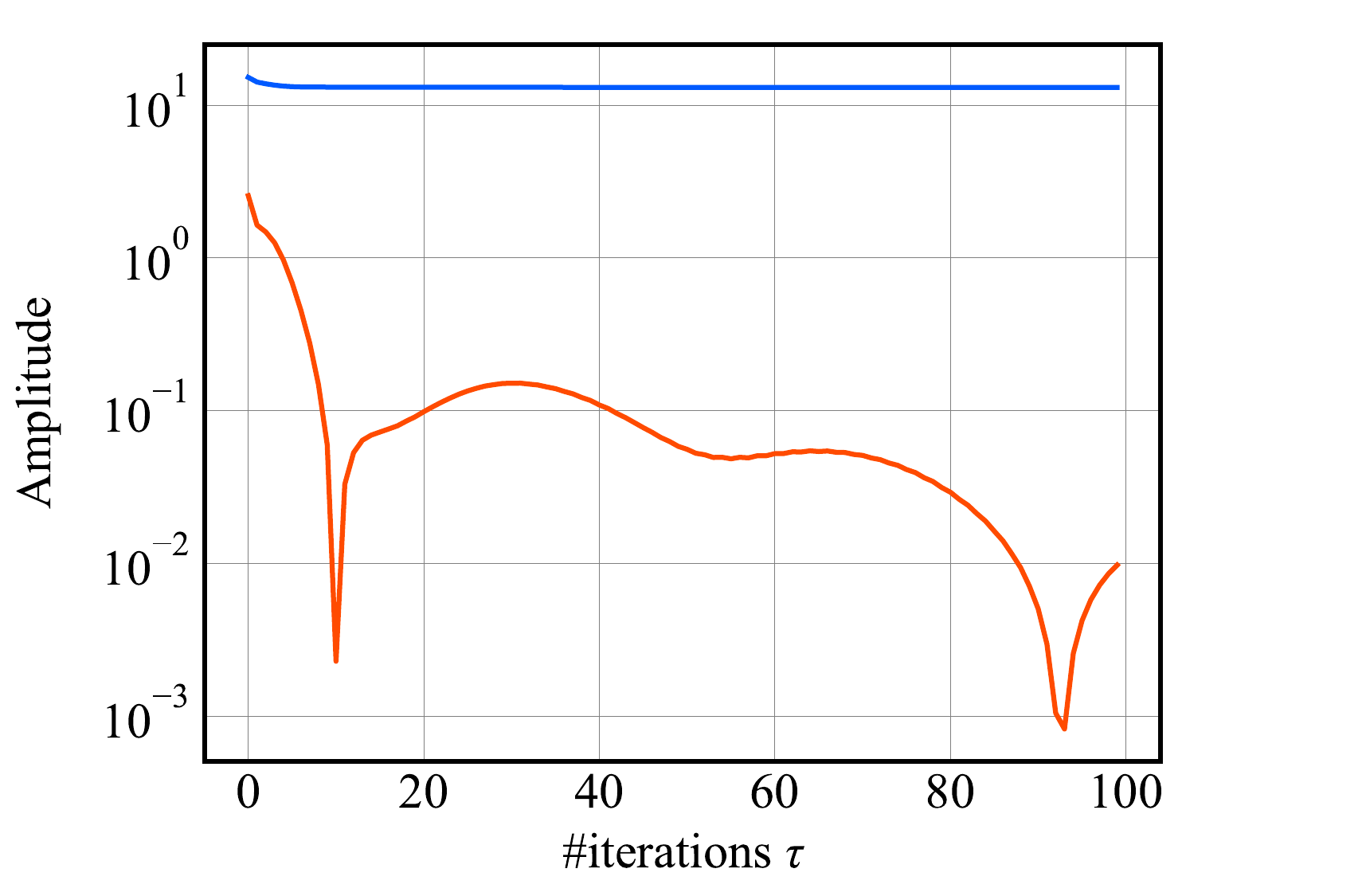}
\caption{$|\partial L(\boldsymbol{c}[\tau])/\partial \bar{c}_n|$~(red) and its upper bound $1/|c_n[\tau]|$~(blue)}
\label{fig:thorem_b}
\vspace{3.5mm}
\end{subfigure}\\
\centering
\begin{subfigure}{0.99\textwidth}
\centering
\includegraphics[width=0.49\textwidth]{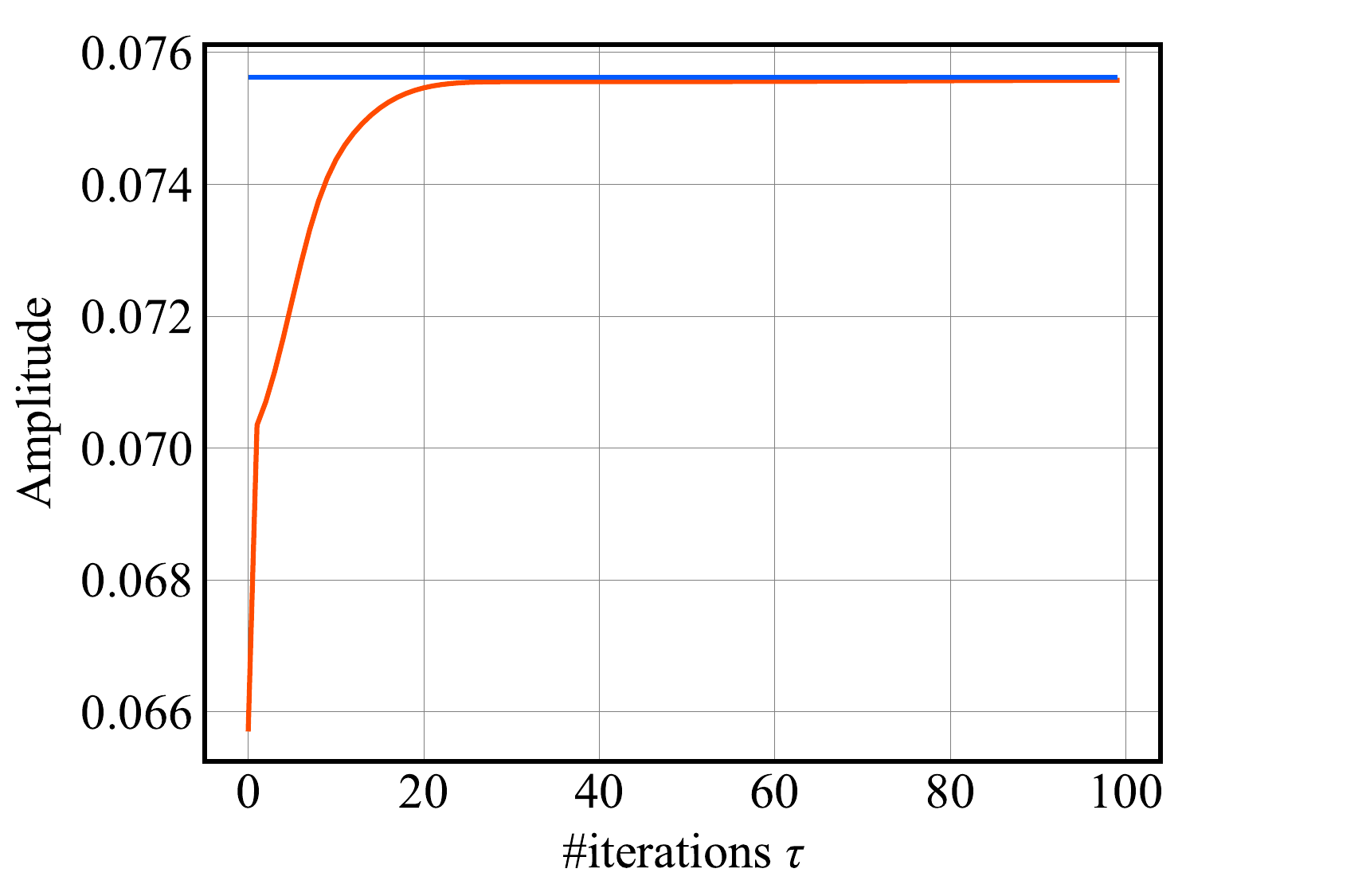}\hspace{0.7mm}
\includegraphics[width=0.49\textwidth]{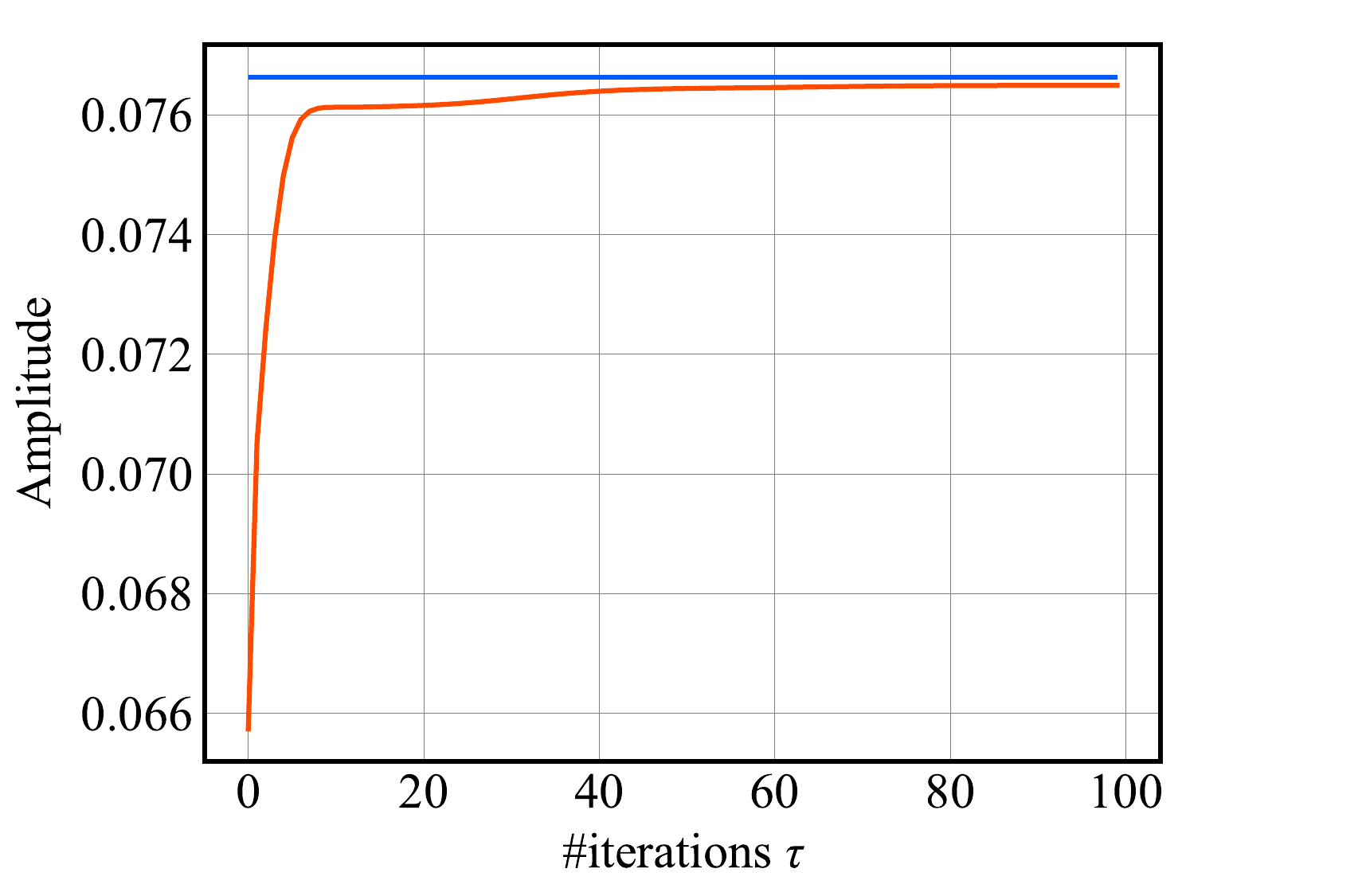}
\caption{$|c_n[\tau]|$~(red) and its limit $|c_n[\infty]|$~(blue)}
\label{fig:thorem_c}
\vspace{3.5mm}
\end{subfigure}\\
\begin{subfigure}{0.99\textwidth}
\centering
\includegraphics[width=0.49\textwidth]{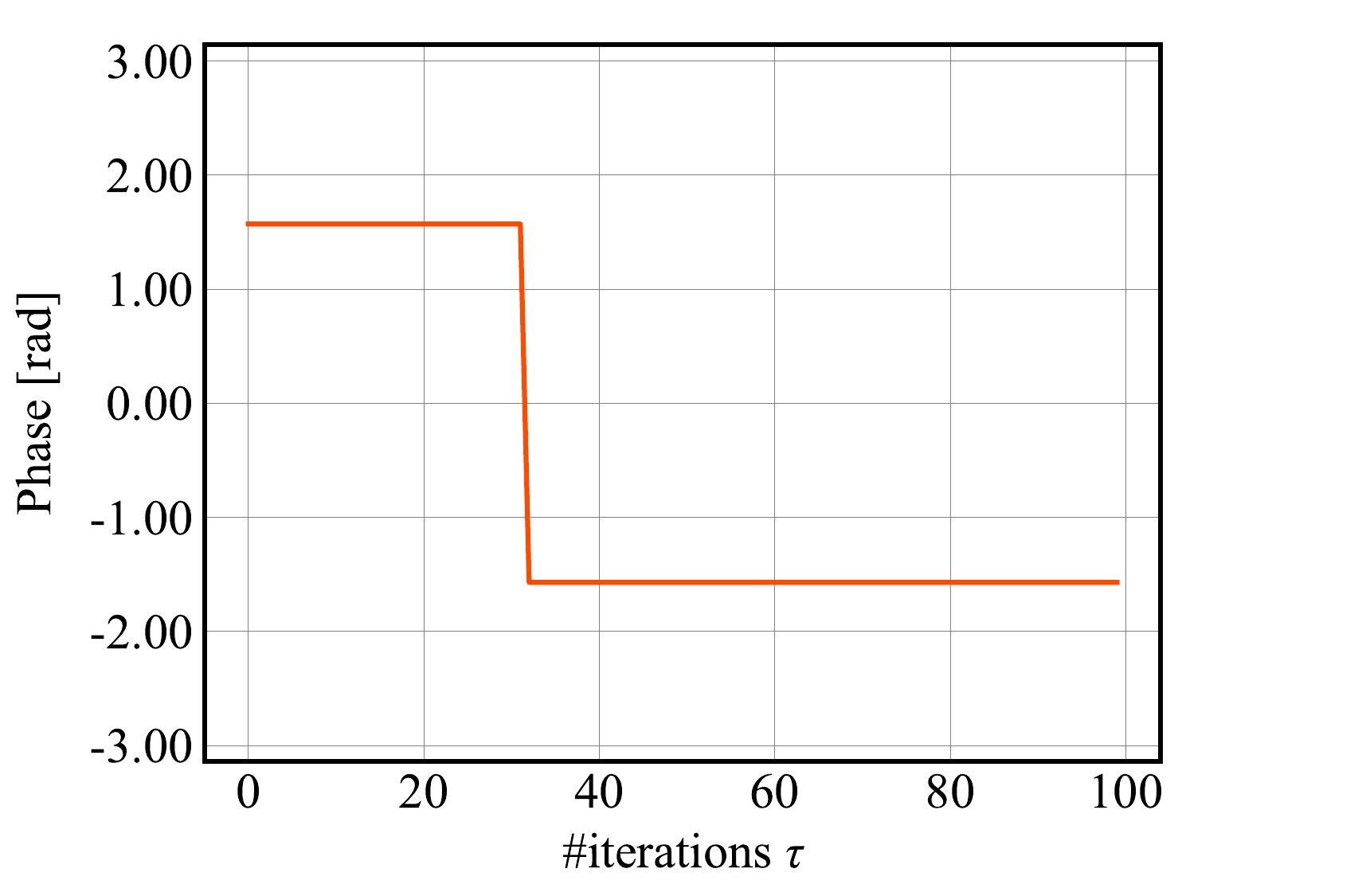}\hspace{0.7mm}
\includegraphics[width=0.49\textwidth]{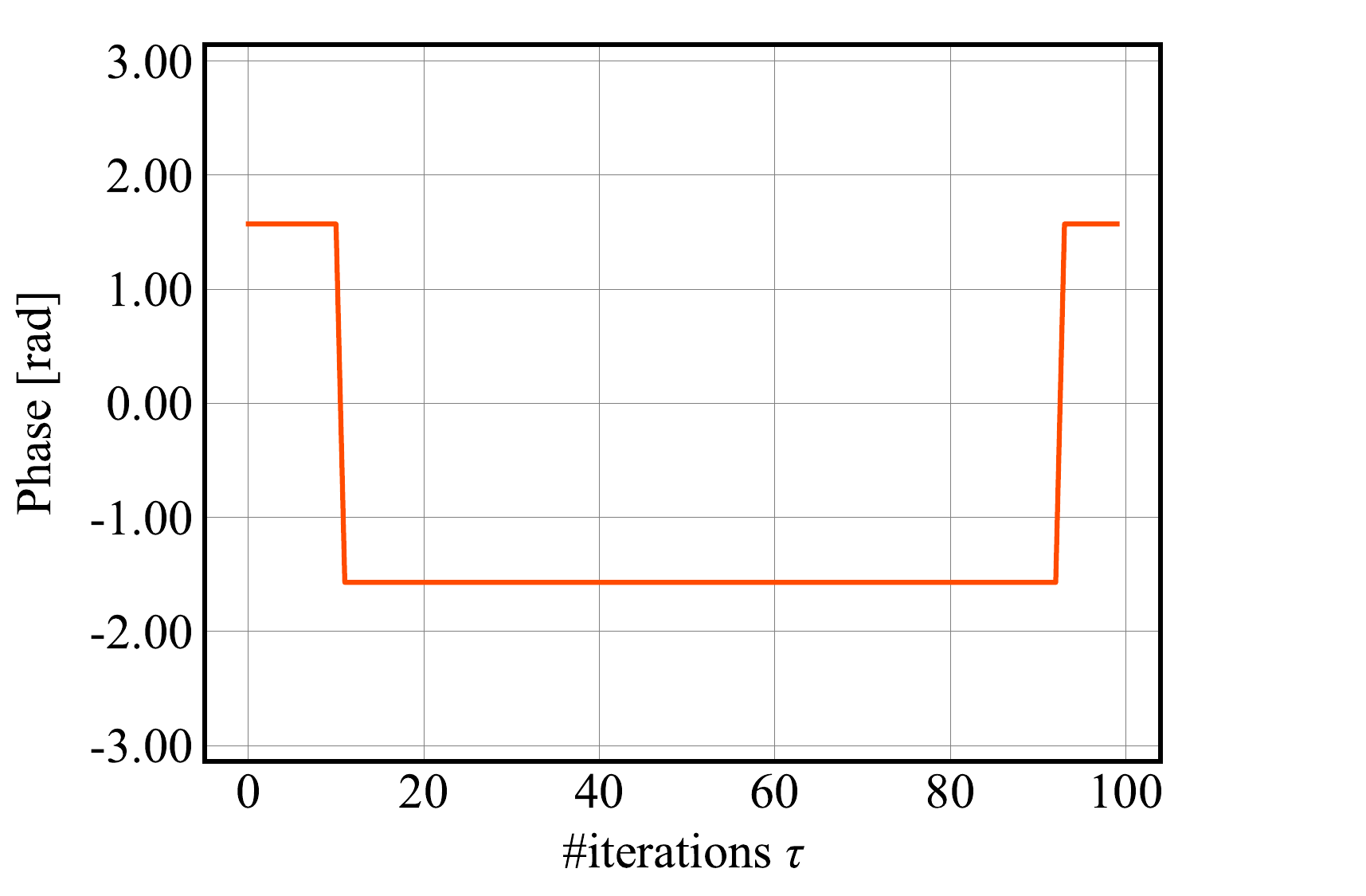}
\caption{$\mathop{\mathrm{arg}}(\partial L(\boldsymbol{c}[\tau])/\partial \bar{c}_n) - \mathop{\mathrm{arg}}(c_n[\tau])$}
\label{fig:thorem_d}
\vspace{1.5mm}
\end{subfigure}
\caption{Amplitudes and phases~(left: synthetic, right: \textit{Rock}).}
\label{fig:thorem}
\end{figure}

\clearpage
\appendix
\def\thesection{Appendix~\Alph{section}}
\section{Proof of Theorem~\ref{t1}}
\label{a1}
First, we claim the following two statements.
\begin{lemma}
\label{l1}
\begin{equation}
\frac{\partial}{\partial \bar{c}_n} \frac{c_n}{|c_n|}= 
-\frac{1}{2}\frac{1}{|c_n|}\bigg(\frac{c_n}{|c_n|}\bigg)^2
\label{eq:lemma1}
\end{equation}
\end{lemma}

\begin{proof}
For notation convenience, we use $c=a+jb$ instead of $c_n=a_n+jb_n$.
\begin{align}
\frac{\partial}{\partial \bar{c}} \frac{c}{|c|}
&= \frac{1}{2}\bigg(\frac{\partial }{\partial a} \frac{a + jb}{\sqrt{a^2 + b^2}}
+j \frac{\partial }{\partial b} \frac{a + jb}{\sqrt{a^2 + b^2}}\bigg)\\
&= \frac{1}{2}\bigg(\frac{1}{a^2 + b^2} 
\bigg( \sqrt{a^2+b^2} - a\frac{a + jb}{\sqrt{a^2+b^2}} \bigg)
+j \frac{1}{a^2 + b^2} 
\bigg(j\sqrt{a^2+b^2} - b\frac{a + jb}{\sqrt{a^2+b^2}} \bigg)\bigg)\\
&= \frac{1}{2}\frac{1}{a^2 + b^2} 
\bigg(- \frac{(a + jb)(a+jb)}{\sqrt{a^2+b^2}} \bigg)
\label{eq:l1}
\end{align}
Applying $a^2 + b^2 = |c|^2$, $(a + jb)(a+jb)=c^2$,
and $\sqrt{a^2 + b^2} = |c|$ to \eqref{eq:l1}, we derive \eqref{eq:lemma1}.
\end{proof}
\begin{lemma}
\label{l2}
\begin{equation}
\frac{\partial}{\partial \bar{c}_n} \frac{\bar{c}_n}{|\bar{c}_n|} = 
\frac{1}{2}
\frac{1}{|c_n|}
\label{eq:lemma2}
\end{equation}
\end{lemma}

\begin{proof}
\begin{align}
\frac{\partial}{\partial \bar{c}}\frac{\bar{c}}{|\bar{c}|}
&= \frac{1}{2}\bigg(\frac{\partial }{\partial a} \frac{a - jb}{\sqrt{a^2 + b^2}}
+j \frac{\partial }{\partial b} \frac{a - jb}{\sqrt{a^2 + b^2}}\bigg)\\
&= \frac{1}{2}\bigg(\frac{1}{a^2 + b^2} 
\bigg( \sqrt{a^2+b^2} - a\frac{a - jb}{\sqrt{a^2+b^2}}\bigg)
-j \frac{1}{a^2 + b^2} \bigg(j\sqrt{a^2+b^2} + b\frac{a - jb}{\sqrt{a^2+b^2}}\bigg)\bigg)\\
&= \frac{1}{2}\frac{1}{a^2 + b^2} 
\bigg(- \frac{(a - jb)(a + jb)}{\sqrt{a^2+b^2}}+2 \sqrt{a^2 + b^2}\bigg)
\label{eq:l2}
\end{align}
Applying $a^2 + b^2 = |c|^2$, $(a - jb)(a+jb)=|c|^2$,
and $\sqrt{a^2 + b^2} = |c|$ to \eqref{eq:l2}, we derive \eqref{eq:lemma2}.
\end{proof}

We are then dedicated to computing the Wirtinger derivative of $L$ with respect to $\bar{c}_n$. 
In accordance with the chain rule~\cite{Wirtinger27}, we derive the following equation:
\begin{align}
\frac{\partial}{\partial \bar{c}_n}L=\frac{1}{2N^2}
\sum_{m=0}^{M-1} \bigg(|p_m(\boldsymbol{c})|^2 - |i_m|^2\bigg)
\frac{\partial}{\partial \bar{c}_n}|p_m(\boldsymbol{c})|^2.
\label{eq:g}
\end{align}
Applying the product rule~\cite{Wirtinger27},
we obtain the derivative of $|p_m(\boldsymbol{c})|^2$ as follows:
\begin{align}
\frac{\partial}{\partial \bar{c}_n}|p_m(\boldsymbol{c})|^2
&=\frac{\partial}{\partial \bar{c}_n}
\overline{p_m(\boldsymbol{c})}p_m(\boldsymbol{c})
\label{eq:d1}\\
&=\overline{p_m(\boldsymbol{c})}
\frac{\partial}{\partial \bar{c}_n}p_m(\boldsymbol{c})+
p_m(\boldsymbol{c})
\frac{\partial}{\partial \bar{c}_n}\overline{p_m(\boldsymbol{c})}.
\label{eq:d1c}
\end{align}
Derivatives in \eqref{eq:d1c} can be represented as the following summation formulas:
\begin{align}
\frac{\partial}{\partial \bar{c}_n}p_m(\boldsymbol{c})&=
\sum_{n'=0}^{N-1} 
\frac{\partial}{\partial \bar{c}_n} w_{mn'} h_{n'}(c_{n'})
\label{eq:d2}\\
\frac{\partial}{\partial \bar{c}_n}\overline{p_m(\boldsymbol{c})}&=
\sum_{n'=0}^{N-1} 
\frac{\partial}{\partial \bar{c}_n} \overline{w_{mn'} h_{n'}(c_{n'})}.
\label{eq:d2c}
\end{align}
When $n' \neq n$, the derivatives in \eqref{eq:d2} and \eqref{eq:d2c} are nullified; we derive
\begin{align}
\frac{\partial}{\partial \bar{c}_n}p_m(\boldsymbol{c})&=
\frac{\partial}{\partial \bar{c}_n} w_{mn} h_n(c_n)=
w_{mn} \frac{\partial}{\partial \bar{c}_n} \frac{c_n}{|c_n|}
\label{eq:d3}\\
\frac{\partial}{\partial \bar{c}_n}
\overline{p_m(\boldsymbol{c})}&=
\frac{\partial}{\partial \bar{c}_n} \overline{w_{mn} h_n(c_n)}=
\bar{w}_{mn} \frac{\partial}{\partial \bar{c}_n} 
\frac{\bar{c}_n}{|\bar{c}_n|}.
\label{eq:d3c}
\end{align}
Applying Lemmas~\ref{l1} and \ref{l2} to \eqref{eq:d3} and \eqref{eq:d3c}, respectively, we derive
\begin{align}
\text{RHS of \eqref{eq:d3}}&=
-\frac{1}{2}w_{mn}\frac{1}{|c_n|}\bigg(\frac{c_n}{|c_n|}\bigg)^2
\label{eq:d4}\\
\text{RHS of \eqref{eq:d3c}}&=
\frac{1}{2}\bar{w}_{mn}\frac{1}{|c_n|}.
\label{eq:d4c}
\end{align}
By substituting \eqref{eq:d1}--\eqref{eq:d4c} to \eqref{eq:g}, we derive
\begin{align}
\frac{\partial}{\partial \bar{c}_n}L&=\frac{1}{4N^2}
\sum_{m=0}^{M-1} \bigg(|p_m(\boldsymbol{c})|^2 - |i_m|^2\bigg) 
\bigg(
-\overline{p_m(\boldsymbol{c})} 
w_{mn} \frac{1}{|c_n|}\bigg(\frac{c_n}{|c_n|}\bigg)^2 +
p_m(\boldsymbol{c}) \bar{w}_{mn} \frac{1}{|c_n|}
\bigg)\\
&=\frac{1}{4N^2}\frac{1}{|c_n|}\frac{c_n}{|c_n|}
\sum_{m=0}^{M-1} 
\underbrace{
\bigg(
|p_m(\boldsymbol{c})|^2 - |i_m|^2
\bigg)}_{R_m}
\bigg(
-\underbrace{\overline{p_m(\boldsymbol{c})} w_{mn} \frac{c_n}{|c_n|}}_{\bar{I}_m} +
\underbrace{p_m(\boldsymbol{c}) \bar{w}_{mn} 
\frac{\bar{c}_n}{|\bar{c}_n|}}_{I_m}
\bigg).
\label{eq:g'}
\end{align}
Note that $R_m$ is a real number, 
and that $I_m-\bar{I}_m$ is a pure imaginary number because $\bar{I}_m$ and $I_m$ have a conjugate relationship.
Using these facts and several fundamental properties of the complex numbers, 
we obtain the following identity:
\begin{equation}
\sum_{m=0}^{M-1} R_m (I_m-\bar{I}_m) = 
2j\sum_{m=0}^{M-1} R_m \Im [I_m] =
2j\sum_{m=0}^{M-1} \Im [R_m I_m] = 
2j\Im \bigg[\sum_{m=0}^{M-1} R_m I_m\bigg],
\label{eq:id}
\end{equation} 
where $\Im$ represents the imaginary part of a complex variable.
Applying \eqref{eq:id} to \eqref{eq:g'}, we derive
\begin{align}
\text{RHS of \eqref{eq:g'}}&=\frac{j}{2N^2}\frac{1}{|c_n|}\frac{c_n}{|c_n|}
\Im \bigg[
\sum_{m=0}^{M-1} \bigg(|p_m(\boldsymbol{c})|^2 - |i_m|^2\bigg)
p_m(\boldsymbol{c}) \bar{w}_{mn} 
\frac{\bar{c}_n}{|\bar{c}_n|}\bigg].
\label{eq:grad}
\end{align}
Equation \eqref{eq:grad} shows that the phase of $\partial L/\partial \bar{c}_n$ depends on $j c_n/|c_n|$ and the sign of the imaginary part.
Therefore, the phase of $\partial L/\partial \bar{c}_n$ is 
\begin{equation}
\mathop{\mathrm{arg}} \bigg(\frac{\partial}{\partial \bar{c}_n}L\bigg)=
\mathop{\mathrm{arg}} (\pm jc_n).
\end{equation}
Note that $\pm j$ rotates $c_n$ by $\pm\pi/2$ on the Gaussian plane.
This fact indicates Theorem~\ref{t1}.

\section{Proof of Theorem~\ref{t2}}
\label{a2}
Equation~\eqref{eq:grad} shows that the amplitude of $\partial L / \partial \bar{c}_n$ depends on $1/2N^2$,$1/|c_n|$, 
and the imaginary part.
Therefore, 
\begin{align}
\bigg|\frac{\partial}{\partial \bar{c}_n}L\bigg| &=
\frac{1}{2N^2}
\frac{1}{|c_n|} 
\bigg| \Im \bigg[
\underbrace{\vphantom{\sum_{m=0}^{M-1}}
\frac{\bar{c}_n}{|\bar{c}_n|}}_{\alpha}
\underbrace{\sum_{m=0}^{M-1} \bigg(|p_m(\boldsymbol{c})|^2 - |i_m|^2\bigg) 
p_m(\boldsymbol{c}) \bar{w}_{mn}}_{\beta} \bigg]\bigg|.
\label{eq:g_abs}
\end{align}
For a complex scalar $\alpha\beta$, the following inequality holds:
\begin{equation}
|\Im [\alpha\beta]| \leq |\alpha \beta| = 
|\alpha| |\beta| = |\beta|.
\label{eq:ineq}
\end{equation}
Applying \eqref{eq:ineq} to \eqref{eq:g_abs}, we derive
\begin{align}
\text{RHS of \eqref{eq:g_abs}}
&\leq \frac{1}{2N^2}\frac{1}{|c_n|} \bigg| 
\sum_{m=0}^{M-1} 
\bigg(|p_m(\boldsymbol{c})|^2 - |i_m|^2\bigg) 
p_m(\boldsymbol{c}) \bar{w}_{mn}\bigg|\\
&=\hspace{0.4mm} \frac{1}{2N^2}\frac{1}{|c_n|} \bigg| 
\sum_{m=0}^{M-1} 
|p_m(\boldsymbol{c})|^2 p_m(\boldsymbol{c}) 
\bar{w}_{mn} - 
\sum_{m=0}^{M-1} 
|i_m|^2 p_m(\boldsymbol{c}) \bar{w}_{mn}\bigg|.
\label{eq:g_abs2}
\end{align}
As mentioned in Section~\ref{s2},
when the propagation distance is infinite, 
$w_{mn}$ can be interpreted as the Fourier basis function.
Therefore, $p_m(\boldsymbol{c})$ equals the Fourier transform~(FT) of $h_n(c_n)$.
Let $\boldsymbol{h}$ denotes a vector formed by stacking all $h_n(c_n)$.
In accordance wtih the Wiener-Khinchin theorem~\cite{Wiener30,Khintchine34},
we can interpret $|p_m|^2$ as the FT of the autocorrelation of $\boldsymbol{h}$.
Namely,
\begin{equation}
|p_m|^2 =\sum_{n=0}^{N-1} w_{mn} (\bar{\boldsymbol{h}} * \boldsymbol{h})_n,
\label{eq:ac1}
\end{equation}
where $*$ is a convolution operator.
Let $\boldsymbol{C}\in \mathbb{C}^N$ be a complex vector that satisfies $|i_m|^2 =|p_m(\boldsymbol{C})|^2$,
and $\boldsymbol{H}$ denote a vector with all elements $h_n(C_n)$.
Similar to \eqref{eq:ac1},
we can interpret $|i_m|^2$ as the FT of the autocorrelation of $\boldsymbol{H}$, i.e.,
\begin{equation}
|i_m|^2 =\sum_{n=0}^{N-1} w_{mn} (\bar{\boldsymbol{H}} * \boldsymbol{H})_n.
\label{eq:ac2}
\end{equation}
In summary, $p_m$, $|p_m|^2$, and $|i_m|^2$ are identical to the FTs of $\boldsymbol{h}$, $\bar{\boldsymbol{h}}*\boldsymbol{h}$, and $\bar{\boldsymbol{H}}*\boldsymbol{H}$, respectively.
Because $\bar{w}_{mn}$ is the inverse Fourier basis function,
the following relationship holds on the basis of the convolution theorem and triangle inequality:
\begin{align}
\text{RHS of \eqref{eq:g_abs2}}&=\hspace{0.4mm}
\frac{1}{2N^2}\frac{1}{|c_n|}
|(\bar{\boldsymbol{h}} * \boldsymbol{h} * \boldsymbol{h})_n-
 (\bar{\boldsymbol{H}} * \boldsymbol{H} * \boldsymbol{h})_n|\\
&\leq
\frac{1}{2N^2}\frac{1}{|c_n|}
|(\bar{\boldsymbol{h}} * \boldsymbol{h} * \boldsymbol{h})_n|+
\frac{1}{2N^2}\frac{1}{|c_n|}
|(\bar{\boldsymbol{H}} * \boldsymbol{H} * \boldsymbol{h})_n|.
\label{eq:g_abs3}
\end{align}
Here, we use the fact that $|(\bar{\boldsymbol{h}} * \boldsymbol{h} * \boldsymbol{h})_n|$ can be bounded by the $\ell_\infty$ norm $\|\bar{\boldsymbol{h}} * \boldsymbol{h} * \boldsymbol{h}\|_\infty$.
In accordance with this fact and the Young's inequality~\cite{Young12},
we derive
\begin{align}
|(\bar{\boldsymbol{h}} * \boldsymbol{h} * \boldsymbol{h})_n| \leq
\|\bar{\boldsymbol{h}} * \boldsymbol{h} * \boldsymbol{h}\|_\infty \leq
\|\bar{\boldsymbol{h}} * \boldsymbol{h}\|_2 \|\boldsymbol{h}\|_2 \leq 
\|\bar{\boldsymbol{h}}\|_1 \|\boldsymbol{h}\|_2 \|\boldsymbol{h}\|_2=N^2.
\label{eq:norm1}
\end{align}
Similarly,
\begin{align}
|(\bar{\boldsymbol{H}} * \boldsymbol{H} * \boldsymbol{h})_n| \leq N^2.
\label{eq:norm2}
\end{align}
Applying \eqref{eq:norm1} and \eqref{eq:norm2} to \eqref{eq:g_abs3}, we derive \eqref{eq:thm2}.

\end{document}